%% file: submission-arxiv.tex
\documentclass[arxiv]{imsart}
\def\submission{1}
\def\issupplement{1} %
\RequirePackage{amsthm,amsmath,amsfonts,amssymb}
\RequirePackage[numbers]{natbib}

\def\withcolors{0}
\def\withnotes{0}
\usepackage{ccanonne}
\usepackage{physics} 

\newcommand{\cWpriv}[1][\priv]{\cW^{{\rm priv}, #1}}
\newcommand{\cWcomm}[1][\numbits]{\cW^{{\rm comm},#1}}

\newcommand{\channel}{\tilde{W}}
\newcommand{\ff}{g}

\newcommand{\sparam}{\tau} %
\newcommand{\nohat}[1]{\vphantom{\hat{#1}}#1} %
\newcommand{\ldim}{r}
\newcommand{\zero}{\mathbf{0}}
\newcommand{\one}[1][]{\mathbf{1}_{#1}}
\newcommand{\numone}[1]{\norm{#1}_+}

\newcommand{\kappavar}{\alpha^2}

\startlocaldefs
\theoremstyle{plain}
\newtheorem{theorem}{Theorem}
\newtheorem{corollary}{Corollary}
\newtheorem{lemma}{Lemma}
\theoremstyle{remark}
\newtheorem{definition}{Definition}
\newtheorem{assumption}{Assumption}
\newtheorem{remark}{Remark}

\newcommand{\saX}{\mathfrak{X}} %
\newcommand{\saY}{\mathfrak{Y}} %

\newenvironment{proofof}[1]{\begin{proof}[\bf Proof of {#1}]}{\end{proof}}
\crefname{assumption}{Assumption}{Assumptions}
\endlocaldefs
\begin{document}

\begin{frontmatter}
	\title{Unified lower bounds for interactive high-dimensional estimation under information constraints}
	\runtitle{Unified lower bounds for interactive estimation}

\begin{aug}
\author[A]{\fnms{Jayadev} \snm{Acharya}\ead[label=e1,mark]{acharya@cornell.edu}},
\author[B]{\fnms{Cl\'ement L.} \snm{Canonne}\ead[label=e2]{clement.canonne@sydney.edu.au}},
\author[A]{\fnms{Ziteng} \snm{Sun}\ead[label=e3,mark]{zs335@cornell.edu}},
\and
\author[C]{\fnms{Himanshu} \snm{Tyagi}\ead[label=e4]{htyagi@iisc.ac.in}}
\address[A]{Cornell University, \printead{e1,e3}}
\address[B]{University of Sydney, \printead{e2}}
\address[C]{Indian Institute of Science, Bangalore, \printead{e4}}
\end{aug}
\begin{abstract}
We consider distributed parameter estimation using interactive protocols 
subject to \emph{local information constraints} such as bandwidth 
limitations, local differential privacy,
and restricted measurements.
We provide a unified framework enabling us to derive a variety of (tight) minimax lower bounds for different parametric families of distributions, both continuous and discrete, under any $\lp[\prm]$ loss. Our lower bound framework is versatile and yields ``plug-and-play'' bounds that are widely applicable to a large range of estimation problems, and, for the prototypical case of the Gaussian family, circumvents limitations of previous techniques. In particular, our approach recovers bounds obtained using data processing inequalities and Cram\'er--Rao bounds, two other alternative approaches for proving lower bounds in our setting of interest.
Further, for the families considered, we complement our lower bounds with matching upper bounds.
\end{abstract}

\ifnum\submission=1

\fi

\end{frontmatter}

\clearpage\pagenumbering{gobble} 
\section{Introduction}
\label{sec:intro}
\input{sec-introduction}
\section{The setup}
\label{sec:setting}
\input{sec-setting}
\section{Main result: The information contraction bound}
\label{sec:general-bound}
\input{sec-mainresult}
\section{An Assouad-type bound}
\label{sec:assouad}
\input{sec-assouad}
\section{Applications}
\label{sec:applications}
\input{sec-applications}

\section*{Acknowledgments}
The authors would like to thank Yanjun Han, and the anonymous reviewers for helpful comments on an earlier version of this paper.

\bibliographystyle{imsart-number} %
\bibliography{references-aos}       %

\begin{appendix}
\section{Fully interactive model}
  \label{app:full:interactive}
  \input{sec-fullinteractive}

\section{A measure change bound}
\label{app:measure:change}
  \input{sec-measure-change}
\section{Upper bounds}
\noindent We now describe and analyze the interactive algorithms for the estimation tasks we consider. 
  \subsection{Product Bernoulli Distributions}
    \label{app:ub:bernoulli}
      \input{sec-bernoulli-ub-sparse-interactive}

  \subsection{Gaussian Mean Estimation}
    \label{app:ub:gaussian}
\input{gaussian-ub}

\end{appendix}

\end{document}

%% file: sec-introduction.tex
We consider the problem of parameter estimation under {\em local information constraints},
where the estimation algorithm \jdnew{has access to} only limited information
about each sample. These constraints can be of various types, including communication
 constraints, where each sample must be described using
 a few (\eg{} constant number of) bits; (local) privacy constraints, where
each sample is obtained from a different user 
and the users seek
 to reveal as little as possible about their specific data;
 as well as many others, \eg{} noisy communication
 channels, or limited types of data access such as linear measurements.
Such problems have received a lot of attention in recent years,
motivated by applications such as data analytics in distributed
systems and federated learning.

Our main focus is on information-theoretic lower bounds for the minimax error rates
(or, equivalently, the sample complexity) of these problems. Several recent works have provided different bounds
that apply to specific constraints or work for specific parametric estimation problems,
sometimes without allowing for interactive protocols.
\newerest{
Indeed, handling interactive protocols is technically challenging, and several results in prior work exhibit flaws in their analysis. In particular, even the most basic Gaussian mean estimation problem using interactive communication remains, quite surprisingly, open.}

We present general, ``plug-and-play'' lower bounds for parametric
estimation under information constraints that can be used for any local information
constraint and allows for \emph{interactive} protocols. %
Our abstract bound requires very simple (and natural) assumptions to hold for the underlying
parametric family; in particular, we do not require technical ``regularity'' conditions
that are common in asymptotic statistics. \amargin{This last sentence is very vague. Do we have many examples in mind for such regularity conditions?}

We apply our general bound to \jdnew{canonical} problems of high-dimensional mean estimation
and distribution estimation, under privacy and communication constraints,
for the entire family of $\lp[\prm]$ loss functions for $\prm\ge1$. In addition, we provide
complementary schemes that show that our lower bounds are tight for most settings of interest.

We present more details about our main results in the next section.

\subsection{Our results}
  \label{ssec:results}
\jdnew{Our main contribution is a general approach to establish lower bounds in distributed information-constrained parameter estimation. The setup is described in detail in~\cref{sec:setting} and is illustrated in~\cref{fig:model}. In short, independent samples $X^\ns=(X_1, \dots, X_\ns)$ are generated from an unknown distribution $\p$ from a parametric family $\cP_\Theta=\{\p_\theta, \theta\in \Theta\}$ of distributions. Only limited information $Y_i$ about datum $X_i$ is available to the algorithm. The goal is to estimate the underlying parameter $\theta$ associated with $\p$. Furthermore, we consider interactive estimation, wherein $Y_i$ can depend on $Y_1,\ldots, Y_{i-1}$.}
Our general lower bound, which we develop in~\cref{sec:general-bound},
takes the following form: Consider a collection of
distributions $\{\p_z\}_{z\in\bool^\zdims}\subseteq \cP_\Theta$ contained
in the parametric family. This collection represents a ``difficult subproblem''
that underlies the parametric estimation problem
being considered; such constructions
are often used when deriving information-theoretic lower bounds.
Note that each coordinate of $z$ represents, in essence, a different ``direction''
of uncertainty for the parameter space. 
The difficulty of the estimation problem can be related to the difficulty
of determining a randomly chosen $z$ \jdnew{(or most of the coordinates of $z$)}, denoted $Z$, by observing samples
from $\p_z$. 
Once $Z=z$ is fixed, $\ns$
independent samples $X^\ns=(X_1, \dots, X_\ns)$ are generated from $\p_z$ and
the limited information $Y_i$ about $X_i$ is passed to an estimator.

Our most general result, stated as~\cref{lemma:per:coordinate}, is an upper bound for the average discrepancy, an average distance quantity related to average
probability of error in determining coordinates of $Z$
by observing the limited information $Y^\ns=(Y_1, \dots, Y_\ns)$. 
Our bounding term reflects the underlying information
constraints using a quantity that captures how ``aligned'' we can make our
information about the sample to the uncertainty in different coordinates of 
$Z$. 
Importantly, our results 
hold under minimal assumptions. In 
particular, in contrast to many previous works, our results do not require 
any 
``bounded ratio'' assumption on the 
collection $\{\p_z\}_{z\in\bool^\zdims}$, which would ask that the density 
function change by at most a constant factor if we modify one coordinate 
of 
$z$.
When we impose additional structure for $\p_z$ -- such
as orthogonality of the random changes in density when we modify
different coordinates of $z$ and, more stringently, independence
and subgaussianity of these changes -- we get concrete bounds
which \jdnew{are readily applicable to} different problems.
These plug-and-play bounds are stated as consequences of our main result in~\cref{thm:avg:coordinate}.

We demonstrate the versatility of the framework by showing that it \newest{readily yields} tight (and in some cases nearly tight) bounds for parameter estimation (both in the sparse and dense cases) for several \newest{fundamental} families of continuous and discrete distributions, several families of information constraints such as communication and local differential privacy (LDP), and for the family of $\lp[\prm]$ loss functions for $\prm\ge1$, all when interactive
protocols are allowed. To complement our lower bounds, we provide algorithms (protocols) which attain the stated rates, thus establishing optimality of our results.\footnote{Up to a logarithmic factor in the case of $\lp[\infty]$ loss, or, for some of our bounds, with a mild restriction on $\ns$ being large enough.}%
We discuss these results in~\cref{sec:applications}; see~\cref{theorem:mean:estimation:bernoulli,theorem:mean:estimation:gaussian,theorem:mean:estimation:discrete} for the corresponding statements.

\jdnew{In terms of the applications, our contributions are two-fold.} 

\jdnew{We obtain several results from a diverse set of prior works in a unified fashion as simple corollaries of our main result: As discussed
further in~\cref{ssec:previous}, the lower bounds for mean estimation
for product Bernoulli under $\lp[2]$ loss and those for estimation of discrete distributions
under $\lp[1]$ and $\lp[2]$ losses, for both communication and LDP
constraints, were known from previous work. However, our approach
allows us to easily recover those results and extend them to arbitrary
$\lp[\prm]$ losses, with interaction allowed, in a unified fashion.}

\jdnew{Our bounds also yield new lower bounds for some canonical problems. The prototypical example where this happens is for mean estimation for high-dimensional Gaussian distributions under information constraints. As discussed in the next section, while some prior work claimed lower bounds for this problem, their arguments appear to be flawed~--~at a high level, due to the ``bounded ratio'' assumption their techniques rely on, which Gaussian distributions do not satisfy, and which our framework does not require. To the best of our knowledge our work is the first to obtain those lower bounds for interactive mean estimation of high-dimensional Gaussian distributions under communication or local privacy constraints.}

We elaborate on prior work in the next section.

\subsection{Previous and related work}
  \label{ssec:previous} There is a significant amount of work in the
literature dedicated to parameter estimation under various
{constraints} and settings. Here, we restrict
our discussion
to works that are most
relevant to the current paper, with a focus on the interactive setting (either
the sequential or blackboard model; see~\cref{sec:setting} for definitions).

The work arguably closest to ours is the recent work~\cite{IIUIC},
which focuses on density estimation and 
goodness-of-fit testing of discrete distributions, under the $\lp[1]$
metric, for sequentially interactive protocols under general local
information constraints (including, as special cases, local privacy
and communication constraints, as in the present paper).
This work can be seen as a significant generalization of
the techniques of~\cite{IIUIC}, allowing us to obtain lower bounds for
estimation in a variety of settings, notably high-dimensional
parameter estimation.

Among other works on high-dimensional mean estimation under
communication constraints, \cite{GargMN14, BGMNW:16}
consider communication-constrained 
Gaussian mean estimation in  the blackboard communication model, under $\lp[2]$
loss. 
\jdnew{The protocols for the upper bounds in these works
do not require interactivity and are complemented with  
 lower bounds which show that the bounds are tight up to constant factors 
 in the dense case and up to logarithmic factors in the sparse case.} 
 \newest{However, the proof of the lower bound in~\cite{BGMNW:16} seems 
 to present a gap (specifically, in the truncation argument 
 of~\cite[Theorem~4.3]{BGMNW:16}), as confirmed in personal 
 communication with the authors.} \newerest{Correcting the issue in the truncation argument would lead to a result  significantly weaker than the claimed lower bound, %
 and it is 
 unclear  whether  this  can  be fixed using the techniques from that paper.}

In this 
 work, we present interactive protocols for the sparse case which improve 
 over the noninteractive protocols and \newest{strenghten the upper 
 bounds by a logarithmic factor in the interactive case} (we elaborate on this 
 in~\cref{rk:interactive:gap}). \newest{Further,  using our general 
 framework, we establish a nearly-matching lower bound for the problem, recovering the 
 rate lower bound originally claimed in~\cite{BGMNW:16} up to a logarithmic factor.}
In a slightly different setting,~\cite{Shamir:14}, among other things, considers
 the mean estimation problem for product Bernoulli distributions when the mean vector
 is $1$-sparse, under $\lp[2]$ loss. The lower bound in~\cite{Shamir:14}, too, allows sequentially
 interactive protocols.
In the blackboard
communication model,~\cite{HOW:18} and~\cite{HMOW-ISIT:18} obtained several tight
bounds for mean estimation and density estimation under $\lp[2]$ and
$\lp[1]$ loss, respectively.

Turning to local privacy,~\cite{JosephKMW19} provide upper bounds (as
well as some partial lower bounds) for one-dimensional Gaussian mean
estimation under LDP under the $\lp[2]$ loss, in the sequentially
interactive model. 
Recent works of~\cite{BHO:19} and~\cite{BCO:20} obtain lower bounds
for mean estimation in the blackboard communication model and under LDP,
respectively, for both Gaussian and product Bernoulli distributions;
 as well as density estimation for discrete distributions.
Their approach is based
 on the classic Cram\'er--Rao bound
and, as such, is tied inherently to the use of
 the $\lp[2]$ loss.  In a recent independent work,~\cite{SZ:21:ISIT} extended these methods to obtain lower bounds under general $\lp[\prm]$ loss under communication constraints, which are tight for Gaussian mean estimation under noninteractive protocols. 
\cite{DJW:17}, by developing a locally private counterpart of some
of the well-known
information-theoretic tools
for establishing statistical lower bounds
(namely, Le Cam, Fano, and
Assouad), establish tight or nearly tight bounds for several mean
estimation problems in the LDP setting. 

More recently, drawing on machinery from the communication complexity
literature,~\cite{DR:19} develop a methodology for proving lower
bounds under LDP constraints in the blackboard communication model.
They obtain lower bounds for mean estimation of product Bernoulli 
distributions
under general $\lp[\prm]$ losses which match ours (in the high-privacy regime, \ie small $\priv$). Similar to the results under 
communication constraints~\cite{GargMN14, BGMNW:16}, their approach 
relies 
heavily on the assumption that the distributions on each coordinate are 
independent, which fails 
to generalize 
to discrete distributions. Moreover, their 
bounds are tailored
to the LDP constraints and do not seem to extend to arbitrary information constraints. \newest{Finally, while~\cite{DR:19} also claims tight
bounds for mean estimation of Gaussian and sparse Gaussian
distributions under the $\lp[2]$ loss, their argument invokes the analogous (flawed) result from~\cite{BGMNW:16} and as such it is unclear whether the stated lower bound can be shown using their techniques.}

Finally, we mention that very recently, following the appearance of~\cite{IIUIC}, an updated version
of~\cite{HOW:18} appeared online as~\cite{HOW:18:v3}, which has
similar results as ours for the high-dimensional mean estimation problem under communication constraints. Both our work and~\cite{HOW:18:v3} build upon the framework presented for the discrete setting in~\cite{IIUIC}.

\subsection{Notation}
  \label{sec:preliminaries}
Hereafter, we write $\log$ and $\ln$ for the binary and natural
logarithms, respectively. For distributions
$\p_1,\p_2$ over $\cX$, denote
their Kullback--Leibler divergence (in nats) by
$
\kldiv{\p_1}{\p_2}
$, 
and their Hellinger distance by
\begin{equation*}
\hellinger{\p_1}{\p_2} \eqdef \sqrt{ \frac{1}{2}\int \Paren{ \sqrt{\dv{\p_1}{\lambda}} - \sqrt{\dv{\p_2}{\lambda}} }^2 \dd{\lambda} }\,,
\end{equation*} 
where we assume $\p_1,\p_2 \ll \lambda$ for some underlying measure 
$\lambda$ on $\cX$.
Further, we denote the Shannon entropy of a random variable $X$ by $H(X)$
and the mutual information between $X$ and $Y$ by $\mutualinfo{X}{Y}$; 
we will sometimes write $H(\p)$ for the entropy of a random variable with 
distribution $\p$. We refer the reader to \cite{CoverThomas:06} for details 
on these notions and their properties, which will be used throughout. Given 
two functions $f,g$, we write $f \lesssim g$ if there exists an absolute 
constant $C>0$ such that $f(x) \leq C g(x)$ for all $x$; and $f\asymp g$ if 
$f \lesssim g$ and $f \gtrsim g$ both hold. Finally, we will extensively use 
the standard asymptotic notation $\bigO{f}$, $\bigOmega{f}$, 
$\bigTheta{f}$.

\subsection{Organization} In \cref{sec:setting}, we formalize our setting of 
interactive inference 
under local information constraints. The general lower bound framework 
and result are 
presented in 
\cref{sec:general-bound}. In \cref{ss:plug-play}, we provide 
implications of the general result under additional structures  and specific 
information constraints. We then compare our 
approach with existing techniques in \cref{sec:relation}. Finally, we use 
our 
framework to derive lower bounds for a wide range of applications in 
\cref{sec:applications}.

%% file: sec-setting.tex
We consider standard parametric estimation problems.
For some $\Theta\subseteq\R^\dims$, let $\cP_\Theta=\{\p_\theta, \theta\in \Theta\}$ be a family of distributions over some measurable space $(\cX, \saX)$, namely each $\p_\theta$ is a distribution over $(\cX, \saX)$. 
Suppose $\ns$ independent samples $X^\ns=(X_1,\ldots, X_{\ns})$
from an unknown $\p_\theta\in \cP_\Theta$ are obtained.
The goal in parametric estimation is to design estimators $\hat{\theta}:\cX^\ns \to \Theta$, and form estimates $\hat{\theta}(X^\ns)$ of $\theta$
using independent samples $X^\ns$ from $\p_\theta$.
We illustrate our results using two specific distribution families: discrete probability mass functions (pmfs) and
high-dimensional product distributions with unknown mean vectors. We will describe the precise minimax setting in detail later in this section.

We are interested in an information-constrained setting, where we do not
have direct access to the samples $X^\ns$ from
$\p_\theta$. Instead, we can only obtain limited information about each
\newest{datapoint $X_i$}. 
Following~\cite{AcharyaCT:IT1}, we model
these information constraints by specifying an allowed set
of \emph{channels} $\cW$ with input alphabet $\cX$ and some output space $\cY$.\footnote{Formally, a channel is a Markov kernel $W\colon \saY\times \cX\to [0,1]$, which we assume to be absolutely continuous with respect to some underlying measure $\mu$ on $(\cY,\saY)$. When clear from the context, we will drop the reference to the $\sigma$-algebras $\saX$ and $\saY$; in particular, in the case of finite $\cX$ or $\cY$. } 
Each sample $X_i$ is passed through a channel from $\cW$, chosen appropriately,
and its output $Y_i$ is the observation we get.
This setting is quite general and captures as special cases the
popular communication and privacy constraints, as we will describe momentarily.

We now formally describe the setting, which is illustrated in~\cref{fig:model}. $\ns$ \iid samples $X_1, \ldots, X_{\ns}$ from an unknown distribution $\p_\theta\in \cP_\Theta$ are observed by players (users) where player $t$
observes $X_t$. Player $t\in [\ns]$ selects a channel $W_t \in \cW$
and sends the message $Y_t$ to a referee, where $Y_t$ is drawn from the probability measure $W_t(\cdot\mid X_t, Y_1,\dots,Y_{t-1})$.
The referee observes $Y^\ns\eqdef (Y_1, \ldots, Y_{\ns})$ and seeks to estimate
the parameter $\theta$.

The freedom allowed in the choice of $W_t$ at the players gives
rise to various communication protocols. We focus on \emph{interactive protocols}, where
channels are chosen by one player at a time, and they can use all
previous messages to make this choice. We describe this class of
protocols below, where we further allow each player $t$ to have a different set of constraints $\cW_t$ (\eg a different communication budget), and $W_t$ must be in $\cW_t$. For simplicity of exposition, and as these already encapsulate most of the difficulties, we focus here on the case of \emph{sequentially} interactive protocols, which has been widely considered in the literature and captures many settings of interest. However, we emphasize that our results extend to the more general class of fully interactive protocols, which we define and address in the Supplement (Appendix A).
\begin{definition}[Sequentially Interactive Protocols]
  \label{def:protocol}
Let $X_1, \dots, X_{\ns}$ be \iid samples from $\p_\theta$, $\theta\in \Theta$.
A \emph{sequentially interactive protocol $\Pi$ using $\cW^\ns=(\cW_1,\dots,\cW_\ns)$} involves
mutually independent random variables $U,U_1,\dots,U_\ns$ (independent 
of the input $X_1, \dots, X_{\ns}$) and
mappings $g_t\colon (U,U_t)
\mapsto W_t\in \cW_t$ for selecting
the channel in round $t\in[\ns]$. In round $t$, player $t$
uses the channel $W_t$ to produce the message (output) $Y_t$
according to the probability distribution $W_t(\cdot\mid X_t,Y_1,\dots,Y_{t-1})$.  The messages
$Y^\ns=(Y_1, \ldots, Y_{\ns})$ received by the referee and the
public randomness $U$ (available to all players) constitute the \emph{transcript} of the protocol $\Pi$; the private randomness $U_1,\dots,U_\ns$ (where $U_t$ is local to player $t$) is not part of the transcript.
\end{definition}
In other words, the channel used by player $t$ is a Markov kernel 
\[
    W_t\colon \saY_t\times \cX\times \cY^{t-1}\to[0,1]\,,
\]
with $\cY_t \subseteq \cY$, and $\saY_t$ being a suitable $\sigma$-algebra on $\cY_t$; and, for player $t\in[\ns]$, the family $\cW_{t}$ captures the possible channels allowed to the player. In what follows, to simplify notation we will equivalently see the channel at player $t$ as a (randomized) mapping $W\colon \cX\times \cY^{t-1} \to \cY$, which on input $x\in\cX$ and the previous $t-1$ messages $y^{t-1}\in\cY^{t-1}$ outputs some $y\in\cY$.  

For concreteness, we now instantiate this definition for the two aforementioned types of information constraints, communication and (local) privacy. 

\paragraph*{Communication constraints} Let $\cY \eqdef \{0,1\}^\ast = \bigcup_{m=0}^\infty \{0,1\}^m$. For $\numbits \geq 1$ and $t\geq 1$, let 
\begin{equation}
    \label{eq:comm}
  \cWcomm[\numbits]\eqdef \{W\colon \cX\times \cY^\ast\to\{0,1\}^\numbits\}
\end{equation}
 be the family of channels with input alphabet $\cX$ and output alphabet the set of all $\numbits$-bit strings. This captures the constraint where the message from each player
can be at most $\numbits$ bits long, and corresponds to the choice $\cW^\ns=(\cWcomm[\numbits],\dots,\cWcomm[\numbits])$. Note that allowing a different communication budget to each player can be done by setting $\cW^\ns=(\cWcomm[\numbits_1],\dots,\cWcomm[\numbits_\ns])$.
  
\paragraph*{Local differential privacy constraints} For $\priv>0$ and $t\geq 1$, a channel $W\colon \cX\times \cY^{t-1}\to\cY$ is \emph{$\priv$-locally differentially private (LDP)}~\cite{EvfimievskiGS:03,DMNS:06,KLNRS:11} if 
\begin{equation}
    \label{eq:ldp}
\sup_{S\in \saY}\sup_{y^{t-1}\in\cY^{t-1}}\frac{W(S\mid x_1, y^{t-1})}{W(S\mid x_2, y^{t-1})} \leq
  e^{\priv}, \quad\forall x_1,x_2\in\cX.\,
\end{equation}
We denote by $\cWpriv[\priv]$ the set of all $\priv$-LDP channels.
For sequentially interactive protocols, the $\priv$-LDP condition is captured by setting $\cW^\ns=(\cWpriv[\priv],\dots,\cWpriv[\priv])$. As before, one can allow different privacy parameters for each player by setting
$\cW^\ns=(\cWpriv[\priv_1],\dots,\cWpriv[\priv_\ns])$.
\smallskip

Finally, we formalize the interactive parametric estimation problem
for the family $\cP_\Theta=\{\p_\theta, \theta\in \Theta\}$.
We consider the problem of estimating $\theta$ under
$\lp[\prm]$ loss. For $\prm\in[1,\infty)$, the $\lp[\prm]$ distance between $u,v\in\R^\dims$ is
\begin{equation}
\lp[\prm](u,v) = \norm{u-v}_\prm = \Paren{\sum_{i=1}^\dims 
\abs{u_i-v_i}^\prm}^{1/\prm}\,.
\nonumber
\end{equation}
This definition extends in a natural way to $\prm=\infty$ by taking the limit.\footnote{\label{ft:holder:logd}Moreover, as a consequence of H\"older's inequality, we
have $\lp[\infty](u,v)\leq \lp[\prm](u,v) \leq \dims^{1/\prm}\lp[\infty](u,v)$ for all $\prm \geq 1$ and $u,v\in\R^\dims$, which implies that for $\prm \eqdef \log \dims$
we have
$\lp[\infty](u,v)\leq \lp[\prm](u,v) \leq 2\lp[\infty](u,v)$. That is,
$\lp[\log \dims](u,v)$ provides a factor-2 approximation of the $\lp[\infty]$ loss. This further extends to $s$-sparse vectors, with a factor $s^{1/\prm}$ instead of $\dims^{1/\prm}$.}
\begin{definition}[Sequentially Interactive Estimates]\label{def:estimate}
Fix $\dims\in \N$ and $\prm\in [0, \infty]$.
Given a family $\cP_\Theta$ of distributions on $\cX$,
with $\Theta\subset \R^\dims$,
an \emph{estimate} for $\cP_\Theta$
consists of a sequentially interactive protocol $\Pi$ with transcript
$(Y^\ns, U)$ and estimator $\hat{\theta}\colon(Y^\ns,
U)\mapsto \hat{\theta}(Y^\ns, U)\in \Theta$.  The referee observes the
transcript $(Y^\ns, U)$ and forms the estimate $\hat{\theta}(Y^\ns,
U)$ of the unknown $\theta$. Further, for $\ns\in \N$ and $\dst\in (0,1)$, $(\Pi, \hat{\theta})$ constitutes an {\em $(\ns, \dst)$-estimator} for $\cP_\Theta$
using $\cW$
under $\lp[\prm]$ loss if
for every $\theta\in \Theta$ the transcript $(Y^\ns, U)$ of $\Pi$ satisfies
\[
\bE{\p_\theta^\ns}{\lp[\prm](\theta,\hat{\theta}(Y^\ns, U))^\prm}^{1/\prm}\leq \dst.
\]
Note that the expectation is over the input $X^\ns\sim \p_\theta^\ns$ for
the protocol $\Pi$ and the randomness of $\Pi$. 
\end{definition}
\noindent We illustrate the setting of sequentially interactive protocols in~\cref{fig:model}.
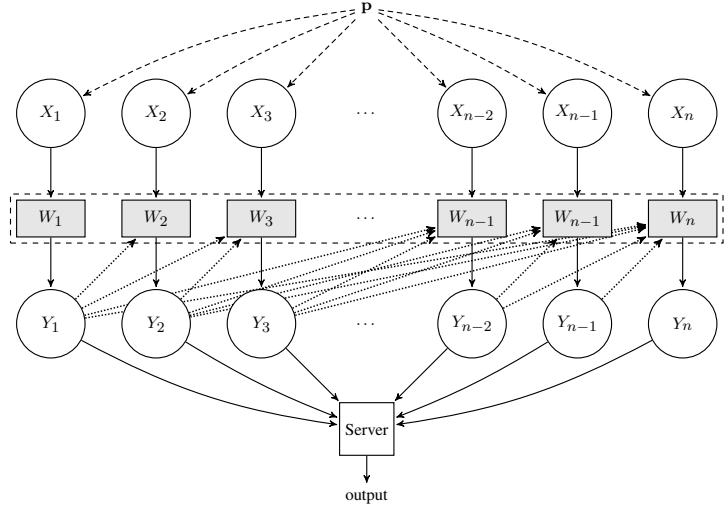
\begin{figure}[ht!]\centering
\input{fig-ic}
\caption{The information-constrained distributed model. In
the interactive setting, $W_t$ can depend on the previous
messages $Y_1, \ldots, Y_{t-1}$ (dotted, upwards arrows).} 
\label{fig:model}
\end{figure}

\begin{remark}
We note that interactive protocols are more general than \emph{simultaneous
message passing} (SMP) protocols, where the channels are chosen by the
players simultaneously, and in particular without the knowledge of
the messages of the other players.
SMP protocols  can themselves be classified into public- and private-coin protocols 
based on whether or not the common random variable $U$ is available to the players; see \cite{AcharyaCT:IT1} for definitions.
\end{remark}
\hnote{We show mention that $g_t$ can be randomized using private randomness. The remark
about public- and private-coin protocols is confusing since even private coin protocols
can have local randomness shared with the center. Can't they?}
\cnote{I have modified the main definition (\cref{def:protocol}) to incorporate private randomness. I think our definition of private randomness precludes sharing with the center: see Definition I.1 in~\cite{AcharyaCT:IT1}, ``The central referee $\mathcal{R}$ does not have access to the realization [...].''}

%% file: fig-ic.tex
\scalebox{0.7}{
\begin{tikzpicture}[->,>=stealth',shorten >=1pt,auto,node distance=20mm, semithick]
  \node[circle,draw,minimum size=13mm] (A) {$X_1$};
  \node[circle,draw,minimum size=13mm] (B) [right of=A] {$X_2$};
  \node[circle,draw,minimum size=13mm] (BB) [right of=B] {$X_3$};
  \node (C) [right of=BB] {$\dots$};
  \node[circle,draw,minimum size=13mm] (DD) [right of=C] {$X_{\ns-2}$};
  \node[circle,draw,minimum size=13mm] (D) [right of=DD] {$X_{\ns-1}$};
  \node[circle,draw,minimum size=13mm] (E) [right of=D] {$X_\ns$};
  
  \node[rectangle,draw,minimum width=13mm,minimum height=7mm,fill=gray!20!white] (WA) [below of=A] {$W_1$};
  \node[rectangle,draw,minimum width=13mm,minimum height=7mm,fill=gray!20!white] (WB) [below of=B] {$W_2$};
  \node[rectangle,draw,minimum width=13mm,minimum height=7mm,fill=gray!20!white] (WBB) [below of=BB] {$W_3$};
  \node[rectangle,minimum width=13mm,minimum height=7mm,fill=none] (WC) [below of=C] {$\dots$};
  \node[rectangle,draw,minimum width=13mm,minimum height=7mm,fill=gray!20!white] (WDD) [below of=DD] {$W_{\ns-2}$};
  \node[rectangle,draw,minimum width=13mm,minimum height=7mm,fill=gray!20!white] (WD) [below of=D] {$W_{\ns-1}$};
  \node[rectangle,draw,minimum width=13mm,minimum height=7mm,fill=gray!20!white] (WE) [below of=E] {$W_\ns$};
  
  \node[draw,dashed,fit=(WA) (WB) (WBB) (WC) (WDD) (WD) (WE)] {};
  \node[circle,draw,minimum size=13mm,fill=white] (YA) [below of=WA] {$Y_1$};
  \node[circle,draw,minimum size=13mm,fill=white] (YB) [below of=WB] {$Y_2$};
  \node[circle,draw,minimum size=13mm,fill=white] (YBB) [below of=WBB] {$Y_3$};
  \node[circle,minimum size=13mm] (YC) [below of=WC] {$\dots$};
  \node[circle,draw,minimum size=13mm,fill=white] (YDD) [below of=WDD] {$Y_{\ns-2}$};
  \node[circle,draw,minimum size=13mm,fill=white] (YD) [below of=WD] {$Y_{\ns-1}$};
  \node[circle,draw,minimum size=13mm,fill=white] (YE) [below of=WE] {$Y_\ns$};
  
  \begin{scope}[on background layer]
  \draw[->] (YA) edge[densely dotted, thick] (WB);
  \draw[->] (YA) edge[densely dotted, thick] (WBB);
  \draw[->] (YA) edge[densely dotted, thick] (WDD);
  \draw[->] (YA) edge[densely dotted, thick] (WE);
  \draw[->] (YB) edge[densely dotted, thick] (WBB);
  \draw[->] (YB) edge[densely dotted, thick] (WDD);
  \draw[->] (YB) edge[densely dotted, thick] (WD);
  \draw[->] (YB) edge[densely dotted, thick] (WE);
  \draw[->] (YBB) edge[densely dotted, thick] (WDD);
  \draw[->] (YBB) edge[densely dotted, thick] (WD);
  \draw[->] (YBB) edge[densely dotted, thick] (WE);
  \draw[->] (YDD) edge[densely dotted, thick] (WD);
  \draw[->] (YDD) edge[densely dotted, thick] (WE);
  \draw[->] (YD) edge[densely dotted, thick] (WE);
  \end{scope}
  
  \node (P) [above of=C] {$\p$};
  \node[rectangle,draw, minimum size=10mm] (R) [below of=YC] {Referee};
  \node (out) [below of=R,node distance=13mm] {output};

  \draw[->] (P) edge[densely dashed,bend right=10] (A)(A) edge (WA)(WA) edge (YA)(YA) edge[bend right=10] (R);
  \draw[->] (P) edge[densely dashed,bend right=5] (B)(B) edge (WB)(WB) edge (YB)(YB) edge[bend right=5] (R);
  \draw[->] (P) edge[densely dashed] (BB)(BB) edge (WBB)(WBB) edge (YBB)(YBB) edge (R);
  \draw[->] (P) edge[densely dashed] (DD)(DD) edge (WDD)(WDD) edge (YDD)(YDD) edge (R);
  \draw[->] (P) edge[densely dashed,bend left=5] (D)(D) edge (WD)(WD) edge (YD)(YD) edge[bend left=5] (R);
  \draw[->] (P) edge[densely dashed,bend left=10] (E)(E) edge (WE)(WE) edge (YE)(YE) edge[bend left=10] (R);
  \draw[->] (R) edge (out);
\end{tikzpicture}
}

%% file: sec-mainresult.tex
Our main result is a unified framework to bound the information revealed 
about the unknown $\theta$ by the transcript of the messages obtained via the constraints defined by the channel family $\cW$. 
The framework is versatile and provides tight bounds for several families of continuous and discrete distributions, several families of information constraints such as communication and local differential privacy, and for the family of $\ell_{\prm}$ loss functions for $\prm\ge 1$. 

Our approach at a high-level proceeds as below: We first consider the ``pertubation space'' $\cZ\eqdef\bool^\zdims$, for some suitable $\zdims$. We associate with each $z\in\cZ$ a parameter $\theta_z\in\Theta$, and refer to $\p_{\theta_z}$ simply as $\p_z$. These distributions are designed in a way that the distance between ${\theta_z}$ and ${\theta_{z'}}$ is large when the Hamming distance between $z$ and $z'$ is large. With this, the difficulty of estimating $\theta$ will be captured in the difficulty of estimating the associated $z$. This will make our approach compatible with the standard Assouad's method for deriving lower bounds ($cf.$~\cite{Yu:97}).

Then, we let $Z=(Z_1, \ldots, Z_\zdims)$ be a random variable
over $\cZ$. Under some assumptions on the distribution of $Z$, we will bound the
information between the individual $Z_i$s and the transcript $(Y^\ns,
U)$ induced by a family of channels $\cW$.  
Combining the two steps above provides us with the desired lower bounds. \medskip

Formally, let $\cZ\eqdef\bool^\zdims$ for some $\zdims$ and
$\{\p_z\}_{z\in\cZ}$ (where $\p_z=\p_{\theta_z}$) be a collection of
distributions over $\cX$, indexed by $z\in\cZ$.  For $z\in\cZ$,
denote by
$z^{\oplus i}\in\cZ$ the vector obtained by flipping the sign of the
$i$th coordinate of $z$. To bound the information that can be obtained
about the underlying $z$ from the observations, we make the following
assumptions:
\begin{assumption}[Densities Exist]
  \label{assn:decomposition-by-coordinates}
For every $z\in\cZ$ and $i\in[\zdims]$ it holds that $\p_{z^{\oplus i}} \ll \p_{z}$, and
there exist measurable functions $\phi_{z,i}\colon\cX\to\R$
such that
\[
 \dv{\p_{z^{\oplus i}}}{\p_z}=1+\phi_{z,i}.
 \]
\end{assumption}
\jdnew{The functions $\phi_{z,i}$ capture the change in density when the coordinate $i$ is flipped.}
In our applications below, we will have discrete distributions or continuous densities, and the Radon--Nikodym derivatives above can be replaced with the corresponding ratios between the pmfs and pdfs, respectively.

\new{
\begin{assumption}[Orthogonality]
  \label{assn:orthonormal}
 	There exists some $\kappavar\geq 0$ such that, for all $z\in\cZ$ and distinct $i,j\in[\zdims]$,
$\bE{\p_{z}}{\phi_{z,i}\phi_{z,j}}=0$ and $\shortexpect_{\p_z}[\phi_{z,i}^2] \leq \kappavar$.
\end{assumption}
}
Note that from~\cref{assn:decomposition-by-coordinates} we
have that $\bE{\p_{z}}{\phi_{z,i}}=0$ for each $i$. In conjunction
with~\cref{assn:orthonormal} this implies that for any fixed
$z\in\cZ$, the family $(1,\phi_{z,1},\dots,\phi_{z,\zdims})$ is
\new{orthogonal and uniformly bounded} in $L^2(\cX,\p_z)$.
Taken together,~\cref{assn:decomposition-by-coordinates} and~\cref{assn:orthonormal}
roughly say that 
the densities can be decomposed into uncorrelated ``perturbations''
across coordinates of $\cZ$.
In later sections, we will show that for several families, such as
discrete distributions, product Bernoulli distributions, and spherical
Gaussians, well-known constructions for lower bounds satisfy these
assumptions.

Our first bound given in~\eqref{eqn:main-bound} only requires~\cref{assn:decomposition-by-coordinates}; by imposing the additional structure of~\cref{assn:orthonormal}, we obtain the more specialized bound given in~\eqref{eqn:var-bound}.
Interestingly,~\eqref{eqn:main-bound}
can be strengthened further when the following
subgaussianity assumption holds.
\begin{assumption}[Subgaussianity]
  \label{assn:subgaussianity} There exists some $\sigma\geq 0$ such
that, for all $z\in\cZ$, the random vector $\phi_z(X)\eqdef
(\phi_{z,i}(X))_{i\in[\zdims]}\in\R^\zdims$ is $\sigma^2$-subgaussian
for $X\sim\p_z$.\footnote{Recall that a random variable $Y$ is
$\sigma^2$-subgaussian if $\bEE{Y}=0$ and \smash{$\shortexpect[e^{\lambda Y}]\leq
e^{\sigma^2\lambda^2/2}$} for all $\lambda\in\R$; and that a
vector-valued random variable $Y$ is $\sigma^2$-subgaussian if its
projection \smash{$\dotprod{Y}{v}$} is \smash{$\sigma^2$}-subgaussian for every unit
vector $v$.}
\end{assumption}

\begin{remark}
  We note that the dimension of the \emph{perturbation} space $\cZ$,
  denoted $\zdims$, need not be the same as the dimension of
  the \emph{parameter} space $\Theta$, which is $\dims$. However, they
  will typically be related, and for many applications in the present
  paper we will have $\zdims=\dims$ or
  $\zdims=\dims/2$.
\end{remark}

Let $Z=(Z_1,\ldots,Z_k)$ be a random variable over $\cZ$ such that  
$\bPr{Z_i=1}=\sparam$ 
for all $i\in[\zdims]$ and the $Z_i$s are all independent;
we denote this distribution by $\rademacher{\sparam}^{\otimes\zdims}$.
Our
main result is an upper bound on the average amount of information
that can be obtained about a coordinate of $Z$ from the transcript
$(Y^\ns, U)$ of a sequentially interactive protocol, as a
function of the information constraint channels and
$\phi_{Z,i}$s.  This result only
requires~\cref{assn:decomposition-by-coordinates}. We then specialize
it to derive~\cref{thm:avg:coordinate}, by
invoking~\cref{assn:orthonormal,assn:subgaussianity}.
\begin{theorem}[Information contraction bound: The technical form]
  \label{lemma:per:coordinate} 
  Fix $\sparam\in(0,1/2]$. Let $\Pi$ be a sequentially interactive
protocol using $\cW^\ns$, and let $Z$ be a random variable on $\cZ$ with distribution $\rademacher{\sparam}^{\otimes\zdims}$. Let $(Y^\ns,U)$ be the transcript of $\Pi$ when the input
$X_1, \ldots, X_\ns$ is i.i.d.\ with common distribution $\p_Z$\zcolor{, 
with density function $\p_Z^{Y^\ns}$}.
Then, under~\cref{assn:decomposition-by-coordinates},
\begin{align}
\label{eqn:main-bound}
&\Paren{\frac{1}{\zdims}\sum_{i=1}^\zdims\totalvardist{\p_{+i}^{Y^\ns}}{\p_{-i}^{Y^\ns}}}^2
\le    
\frac{\new{7}}{\zdims}\sum_{t=1}^\ns\max_{z\in\cZ}\max_{W\in\cW_t} \sum_{i=1}^\zdims \int_{\cY} \frac{\bE{\p_{z}}{\phi_{z,i}(X)W(y\mid X)}^2}{\bE{\p_{z}}{W(y\mid X)}} \dd{\mu} \,, 
\end{align}
where $\p_{+i}^{Y^\ns}\eqdef \bEEC{\p_Z^{Y^\ns}}{Z_i= + 1}$, 
$\p_{-i}^{Y^\ns}\eqdef \bEEC{\p_Z^{Y^\ns}}{Z_i= - 1}$.
\end{theorem}
\begin{remark}[Average discrepancy and mutual information] The  
\emph{average discrepancy}, as defined in 
	the left hand side of \cref{eqn:main-bound}, characterizes the average 
	information the transcript carries about each coordinate of $Z$, which 
	has to be large for any valid parameter estimation protocol as we show in 
	\cref{sec:assouad}.
	It is not hard to show that, 
	borrowing the notation of~\cref{lemma:per:coordinate}, the 
	following 
	inequalities hold:
	\begin{equation}
	\tau\cdot \Paren{\frac{1}{\zdims}\sum_{i=1}^\zdims\totalvardist{\p_{+i}
			^{Y^\ns}}{\p_{-i}^{Y^\ns}}}^2 
	\leq \frac{1}{\zdims} \sum_{i=1}^{\zdims} \mutualinfo{Z_i}{Y^\ns}
	\leq \frac{1}{\zdims} \mutualinfo{Z}{Y^\ns},
	\end{equation}
	which suggests that getting a good upper bound on average mutual 
	information is harder than getting an upper bound on average discrepancy.
\end{remark}
\begin{proofof}{Main result:~\cref{lemma:per:coordinate}}
Consider $Z=(Z_1, \dots, Z_\ab)\in \{-1,1\}^\ab$ where $Z_1, \dots, Z_\ab$ are i.i.d.\ with  $\bPr{Z_i=1}=\sparam$.  
For a fixed $i\in [\zdims]$, let
\begin{align*}
  \p_{+i}^{Y^\ns}\eqdef  \bE{Z}{\p_Z^{Y^\ns}\mid Z_i=+1}  &= 
  \sum_{z:z_i=+1} \Big(\prod_{j\neq i} 
  \sparam^{\frac{1+z_j}{2}}(1-\sparam)^{\frac{1-z_j}{2}}\Big) \p_z^{Y^\ns}  
  \\
  \p_{-i}^{Y^\ns} \eqdef  \bE{Z}{\p_Z^{Y^\ns}\mid Z_i=-1}&= \sum_{z:z_i=-1} \Big(\prod_{j\neq i} \sparam^{\frac{1+z_j}{2}}(1-\sparam)^{\frac{1-z_j}{2}}\Big) \p_z^{Y^\ns} ,
\end{align*}
the partial mixtures of message distributions conditioned on $Z_i$. 
\newest{We will rely on the following lemma, which relates the desired 
average discrepancy between the $\p_{+i}^{Y^\ns}$ and $\p_{-i}^{Y^\ns}$'s to 
the sum of $\ns$ ``local'' discrepancy measures (in the form of Hellinger 
distances between local messages). Each local measure can then be easily 
bounded in terms of the density $\p_z$ and the channel $W$ to get the 
desired bound.
\begin{lemma}\label{lem:main-lemma}
	With the notation of~\cref{lemma:per:coordinate}, we have
	\begin{equation} \label{eqn:tv2hel}
		\Paren{\frac{1}{\zdims}\sum_{i=1}^\zdims\totalvardist{\p_{+i}^{Y^\ns}}{\p_{-i}^{Y^\ns}}}^2
		\le    
		\frac{14}{\zdims}\sum_{t=1}^\ns\max_{z\in\cZ}\max_{W\in\cW_t}
		 \sum_{i=1}^\zdims \hellinger{\p_z^W}{\p_{z^{\oplus i}}^W}^2,
	\end{equation}
	where $\p_z^W$ denotes the distribution of $Y \sim W(\cdot \mid X)$ 
	when $X \sim \p_z$.
\end{lemma}
The proof of the lemma is rather involved and constitutes the core of the argument. We defer it to the end of the section and show first how it implies~\cref{lemma:per:coordinate}. For all $z$ and $W$, we have
\begin{align}
	\hellinger{\p_z^W}{\p_{z^{\oplus i}}^W}^2 & = \frac12 \int_{y \in 
	\cY}\Paren{ 
	\sqrt{\bE{\p_{z}}{W(y \mid X)}} - \sqrt{\bE{\p_{z^{\oplus 
	i}}}{W(y \mid X)}} }^2 \dd{\mu} \nonumber \\
	&= \frac12\int_{\cY} \Paren{ 
	\frac{\bE{\p_{z}}{W(y \mid X)} - \bE{\p_{z^{\oplus i}}}{W(y \mid X)} 
	}{\sqrt{\bE{\p_{z}}{W(y \mid X)}} + \sqrt{\bE{\p_{z^{\oplus i}}}{W(y \mid 
	X)}}} }^2 \dd{\mu} \notag\\
	&\leq \frac12 \int_{\cY} \frac{ 
	(\bE{\p_{z}}{W(y \mid X)} - \bE{\p_{z^{\oplus i}}}{W(y \mid X)})^2 
	}{\bE{\p_{z}}{W(y \mid X)} } \dd{\mu} \label{eq:chi2bound}.
\end{align}
Moreover, under~\cref{assn:decomposition-by-coordinates}; for 
any $W\in\cW_t$ and $y\in\cY$,
\begin{align*}
\bE{\p_{z^{\oplus i}}}{W(y\mid X)}    
&= \bE{\p_{z}}
{\dv{\p_{z^{\oplus i}}}{\p_{z}}(X)\cdot W(y\mid X)}
= \bE{\p_{z}}
{\left(1+\phi_{z,i}(X)\right)
	\cdot
	W(y\mid X)
}\,. 
\end{align*}
Plugging this back into~\eqref{eq:chi2bound}, we get 
\[
	\hellinger{\p_z^W}{\p_{z^{\oplus i}}^W}^2 \le \frac12
	\int_{\cY} \frac{\bE{\p_{z}}{\phi_{z,i}(X)W(y\mid 
	X)}^2}{\bE{\p_{z}}{W(y\mid X)}} \dd{\mu}. 
\]
Combining this with \cref{lem:main-lemma} concludes the proof of~\cref{lemma:per:coordinate}.}
\end{proofof}

\begin{proofof}{\cref{lem:main-lemma}}
Our first step is to use the Cauchy--Schwarz inequality, followed by an 
inequality relating total variation and Hellinger distances:
\begin{align}
  \frac{1}{\zdims}\Paren{\sum_{i=1}^\zdims \totalvardist{\p_{+i}^{Y^\ns}}{\p_{-i}^{Y^\ns}}}^2
  &\leq \sum_{i=1}^\zdims \totalvardist{\p_{+i}^{Y^\ns}}{\p_{-i}^{Y^\ns}}^2 \notag\\
  &\leq 2\sum_{i=1}^\zdims \hellinger{\p_{+i}^{Y^\ns}}{\p_{-i}^{Y^\ns}}^2 \notag\\
  &\leq 2\sum_{i=1}^\zdims \bE{Z}{\hellinger{\p_Z^{Y^\ns}}{\p_{Z^{\oplus i}}^{Y^\ns}}^2\mid Z_i=+1} \notag\\
  &= 2\sum_{i=1}^\zdims \bE{Z}{\hellinger{\p_Z^{Y^\ns}}{\p_{Z^{\oplus i}}^{Y^\ns}}^2 } \label{eq:symmetry}\,,
\end{align}
where the last inequality uses joint convexity of squared Hellinger distance, and the final
identity
is due to independence of each coordinate of $Z$ 
and symmetry of Hellinger whereby $\bE{Z}{\hellinger{\p_Z^{Y^\ns}}{\p_{Z^{\oplus i}}^{Y^\ns}}^2\mid Z_i=+1}=\bE{Z}{\hellinger{\p_Z^{Y^\ns}}{\p_{Z^{\oplus i}}^{Y^\ns}}^2\mid Z_i=-1}$.

In order to bound the resulting terms of the sum, we will rely on the so-called \emph{cut-paste} property of Hellinger distance~\cite{BarYossefJKS04}. Before doing so, we will require an additional piece of notation: for fixed $z \in \cZ$, $i \in [\zdims], t \in [\ns]$, let 	
			$\p^{Y^\ns}_{t\gets z^{\oplus i}}$ denote the message distribution where player $t$ gets a 
					sample from $\p_{z^{\oplus i}}$ and all other players get samples from $\p_{z}$. That is, for all $y^\ns \in \cY^\ns$, the density of $\p^{Y^\ns}_{t\gets z^{\oplus i}}$ with respect to the underlying product measure $\mu^{\otimes\ns}$ is given by
				\begin{equation}\label{eq:ptiy}
					\dv{\p^{Y^\ns}_{t\gets z^{\oplus i}}}{\mu^{\otimes\ns}}\/(y^\ns)= \bE{X_t\sim\p_{z^{\oplus i}}}{W^{y^{t-1}}(y_t \mid X_t)} \cdot \prod_{j \neq t} \bE{X_j\sim\p_{z}}{W^{y^{j-1}}(y_j \mid X_j)}.
				\end{equation}
The following lemma, due to~\cite{Jayram09}, allows us to relate $\hellinger{\p_z^{Y^\ns}}{\p_{z^{\oplus i}}^{Y^\ns}}$, the distance between message distributions when all players get observations from $\p_z$, or all from $\p_{z^{\oplus i}}$, to the distances $\hellinger{\p_z^{Y^\ns}}{\p^{Y^\ns}_{t\gets z^{\oplus i}}}$ where only \emph{one} of the $\ns$ players gets a sample from $\p_{z^{\oplus i}}$.
\begin{lemma}[{\cite[Theorem~7]{Jayram09}}] \label{lem:hellinger:oneoff}
	There exists $c_{\rm H} > 0$ such that for all $z \in \cZ$ and $i \in [\zdims]$,
	\[
	\hellinger{\p^{Y^\ns}_{z}}{\p^{Y^\ns}_{z^{\oplus i}}}^2 \leq c_{\rm H}\sum_{t=1}^\ns \hellinger{\p^{Y^\ns}_{z}}{\p^{Y^\ns}_{t \gets z^{\oplus i}}}^2.
	\]
	Moreover, one can take $c_{\rm H} = 2\prod_{t=1}^\infty 
	\frac{1}{1-2^{-t}} < 7$.
\end{lemma}

Combining~\cref{eq:symmetry,lem:hellinger:oneoff}, we get
\begin{align}
  \frac{1}{\zdims}\Paren{\sum_{i=1}^\zdims \totalvardist{\p_{+i}^{Y^\ns}}{\p_{-i}^{Y^\ns}}}^2
  &\leq 14\sum_{i=1}^\zdims \sum_{t=1}^\ns  \bE{Z}{  \hellinger{\p^{Y^\ns}_{Z}}{\p^{Y^\ns}_{t \gets Z^{\oplus i}}}^2  } \notag \\
  &= 14 \sum_{t=1}^\ns \bE{Z}{ \sum_{i=1}^\zdims \hellinger{\p^{Y^\ns}_{Z}}{\p^{Y^\ns}_{t \gets Z^{\oplus i}}}^2 }.  \label{eq:decomposition}
\end{align}
In view of bounding the RHS of~\eqref{eq:decomposition} term by term, fix $j\in[\ns]$ and $z\in\cZ$. Recalling the expression of $\p^{Y^\ns}_{t\gets z^{\oplus i}}$ from~\eqref{eq:ptiy}, unrolling the definition of Hellinger distance, and recalling~\eqref{eq:ptiy}, we have
\cmargin{Reminder: Hellinger is $\frac{1}{\sqrt{2}}\normtwo{\sqrt{\p}-\sqrt{\q}}$.}
\begin{align*}
  2&\sum_{i=1}^\zdims \hellinger{\p^{Y^\ns}_{z}}{\p^{Y^\ns}_{t \gets z^{\oplus i}}}^2 \\
  &= \sum_{i=1}^\zdims \int_{\cY^\ns} \Paren{ \sqrt{ \dv{\p^{Y^\ns}_{z\vphantom{t\gets z^{\oplus i}}}}{\mu^{\otimes\ns}}  } - \sqrt{\dv{\p^{Y^\ns}_{t\gets z^{\oplus i}}}{\mu^{\otimes\ns}}} }^2 \dd{\mu^{\otimes \ns}} \\
  &= \sum_{i=1}^\zdims \int_{\cY^\ns} \prod_{j \neq t} \bE{\p_{z}}{W^{y^{j-1}}(y_j \mid X)} \underbrace{\Paren{ \sqrt{\bE{\p_{z}}{W^{y^{t-1}}(y_t \mid X)}} - \sqrt{\bE{\p_{z^{\oplus i}}}{W^{y^{t-1}}(y_t \mid X)}} }^2}_{\eqdef f_{i,t}(y^{t-1},y_t)} \dd{\mu^{\otimes \ns}} \\
  &= \sum_{i=1}^\zdims \int_{\cY^{t-1}} \prod_{j < t}  \bE{\p_{z}}{W^{y^{j-1}}(y_j \mid X)} \int_{\cY}f_{i,t}(y^{t-1},y_t) \int_{\cY^{\ns-t}} \prod_{j > t}  \bE{\p_{z}}{W^{y^{j-1}}(y_j \mid X)} \dd{\mu^{\otimes (t-1)}}\dd{\mu}\dd{\mu^{\otimes (\ns-t)}} \\
  &= \sum_{i=1}^\zdims \int_{\cY^{t-1}} \prod_{j < t}  
  \bE{\p_{z}}{W^{y^{j-1}}(y_j \mid X)} \int_{\cY}f_{i,t}(y^{t-1},y_t) 
  \Big(\int_{\cY^{\ns-t}} \prod_{j > t}  \bE{\p_{z}}{W^{y^{j-1}}(y_j \mid X)} 
  \dd{\mu^{\otimes (\ns-t)}} \Big) \dd{\mu^{\otimes (t-1)}}\dd{\mu} \\
  &= \sum_{i=1}^\zdims \int_{\cY^{t-1}} \prod_{j < t}  \bE{\p_{z}}{W^{y^{j-1}}(y_j \mid X)} \int_{\cY}f_{i,t}(y^{t-1},y_t) \dd{\mu} \dd{\mu^{\otimes (t-1)}}\\
  &= \int_{\cY^{t-1}} \prod_{j < t}  \bE{\p_{z}}{W^{y^{j-1}}(y_j \mid X)} \sum_{i=1}^\zdims \int_{\cY}f_{i,t}(y^{t-1},y_t) \dd{\mu} \dd{\mu^{\otimes (t-1)}}\,,
\end{align*}
where the second-to-last identity uses the observation that, for any fixed $y^{t}\in\cY^t$,
\[
  \int_{\cY^{\ns-t}} \prod_{j > t}  \bE{\p_{z}}{W^{y^{j-1}}(y_j \mid X)} \dd{\mu^{\otimes (\ns-t)}}  = 1,
\]
which in turn follows upon taking marginal integrals for each coordinate. 
We then get from the pointwise inequality $\sum_{i=1}^\zdims 
\int_{\cY^{t-1}}f_{i,t}(y^{t-1},y_t) \dd{\mu} \leq \sup_{y'\in \cY^{t-1}} 
\sum_{i=1}^\zdims \int_{\cY} f_{i,t}(y',y_t) \dd{\mu}$ that 
\begin{align}
  2\sum_{i=1}^\zdims \hellinger{\p^{Y^\ns}_{z}}{\p^{Y^\ns}_{t \gets z^{\oplus i}}}^2 
  &\leq \int_{\cY^{t-1}} \prod_{j < t}  \bE{\p_{z}}{W^{y^{j-1}}(y_j \mid X)} \sup_{y'\in \cY^{t-1}} \sum_{i=1}^\zdims \Paren{ \int_{\cY}f_{i,t}(y',y_t) \dd{\mu} } \dd{\mu^{\otimes (t-1)}} \notag\\
  &= \Paren{\sup_{y'\in \cY^{t-1}} \sum_{i=1}^\zdims \int_{\cY}f_{i,t}(y',y_t) \dd{\mu} }  \int_{\cY^{t-1}} \prod_{j < t}  \bE{\p_{z}}{W^{y^{j-1}}(y_j \mid X)} \dd{\mu^{\otimes (t-1)}} \notag\\
  &= \sup_{y'\in \cY^{t-1}} \sum_{i=1}^\zdims  \int_{\cY} \Paren{ \sqrt{\bE{\p_{z}}{W^{y'}(y \mid X)}} - \sqrt{\bE{\p_{z^{\oplus i}}}{W^{y'}(y \mid X)}} }^2 \dd{\mu} \notag\\
  &\leq \sup_{W\in\cW_t} \sum_{i=1}^\zdims \int_{\cY} \Paren{ 
  \sqrt{\bE{\p_{z}}{W(y \mid X)}} - \sqrt{\bE{\p_{z^{\oplus i}}}{W(y \mid X)}} 
  }^2 \dd{\mu} \notag\\
  & = \zcolor{2 \cdot \sup_{W\in\cW_t}  \sum_{i=1}^\zdims 
  \hellinger{\p_z^W}{\p_{z^{\oplus i}}^W}^2.}
\end{align}
the second identity follows upon taking marginal integrals, and by 
replacing $f_{i,t}$ by its definition; and the second inequality using that 
$\setOfSuchThat{W^{y'} }{ y'\in \cY^{t-1} } \subseteq \cW_t$, so that we are 
taking a supremum over a larger set. 

\newest{Plugging this back into~\eqref{eq:decomposition} and upper 
bounding 
the inner expectation by a maximum concludes the 
proof of the lemma.}
\end{proofof}

\section{Implications of the general bound}\label{ss:plug-play}
In stating the general bound above, we have strived to present the general form of our result, suited to applications 
beyond those discussed in the current paper. We now instantiate it to give simple
``plug-and-play'' bounds which can be applied readily to several
inference problems and information constraints.
\begin{theorem}%
  \label{thm:avg:coordinate} 
  Fix $\sparam\in(0,1/2]$. Let $\Pi$ be a sequentially interactive
protocol using $\cW^\ns$, and let $Z$ be a random variable on $\cZ$ with distribution $\rademacher{\sparam}^{\otimes\zdims}$. Let $(Y^\ns,U)$ be the transcript of $\Pi$ when
the input $X_1, \ldots, X_\ns$ is i.i.d.\ with common distribution
$\p_Z$.  Then,
under~\cref{assn:decomposition-by-coordinates,assn:orthonormal},
we have
\begin{align}
\Paren{\frac{1}{\zdims}\sum_{i=1}^\zdims\totalvardist{\p_{+i}^{Y^\ns}}{\p_{-i}^{Y^\ns}}}^2
\le \frac{\new{7}}{\zdims} \kappavar \sum_{t=1}^\ns \max_{z\in\cZ}\max_{W\in\cW_t}\int_{\cY} \frac{\Var_{\p_{z}}[W(y\mid X)]}{\bE{\p_{z}}{W(y\mid X)}} \dd{\mu}.\label{eqn:var-bound}
\end{align}
Moreover, if~\cref{assn:subgaussianity} holds as well, we have
\begin{align}
\Paren{\frac{1}{\zdims}\sum_{i=1}^\zdims\totalvardist{\p_{+i}^{Y^\ns}}{\p_{-i}^{Y^\ns}}}^2 \le \frac{\new{14}}{\zdims}\sigma^2 \sum_{t=1}^\ns \max_{z\in\cZ}\max_{W\in\cW_t}\mutualinfo{\p_z}{W},\label{eqn:subgaussian-bound}
\end{align}
where $\mutualinfo{\p_z}{W}$ denotes the mutual information $\mutualinfo{X}{Y}$ between the input $X\sim \p_z$
and the output $Y$ of the channel $W$ with $X$ as input.
\end{theorem}
As an illustrative and important corollary, we now derive the implications of this theorem for communication and privacy constraints. For both constraints our tight (or nearly tight) bounds in~\cref{sec:applications} follow directly from these corollaries.
\begin{corollary}[Local privacy constraints]
\label{cor:ldp}
For $\cW=\cWpriv[\priv]$ and any family of distributions $\{\p_z, z\in \bool^\zdims\}$
satisfying~\cref{assn:decomposition-by-coordinates,assn:orthonormal},
with the notation of~\cref{thm:avg:coordinate}, we have
\begin{align}
  \label{eq:cor-ldp}
 \Paren{\frac{1}{\zdims}\sum_{i=1}^\zdims\totalvardist{\p_{+i}^{Y^\ns}}{\p_{-i}^{Y^\ns}}}^2 
 \le \frac{\new{7}}{\zdims} \ns\kappavar\Paren{(e^\priv-1)^2 \land e^\priv}.
\end{align}
Moreover, if~\cref{assn:subgaussianity} holds as well, we have
\begin{align}
  \label{eqn:cor-ldp:subgaussian}
 \Paren{\frac{1}{\zdims}\sum_{i=1}^\zdims\totalvardist{\p_{+i}^{Y^\ns}}{\p_{-i}^{Y^\ns}}}^2 
 \le \frac{14}{\zdims}\ns\sigma^2\priv.
\end{align}
\end{corollary}
\begin{proof}
For any $W\in\cWpriv[\priv]$, the $\priv$-LDP condition from~\cref{eq:ldp} can be seen to imply that, for every $y\in\cY$,
\[
        W(y\mid x_1) - W(y\mid x_2) \leq (e^{\priv}-1) W(y\mid x_3), \qquad \forall x_1,x_2,x_3\in\cX\,.
\]
By taking expectation over $x_3$ then
again either over $x_1$ or $x_2$  (all distributed according to $\p_z$), this yields
\[
        \abs{W(y\mid x) - \bE{\p_z}{W(y\mid X)}} \leq
        (e^{\priv}-1) \bE{\p_z}{W(y\mid X)}, \qquad \forall x \in\cX\,.
\]
Squaring and taking the expectation on both sides, we obtain
\[
\Var_{\p_z}[W(y\mid X)] \leq (e^{\priv}-1)^2\, \bE{\p_z}{W(y\mid X)}^2.
\]
Dividing by $\bE{\p_z}{W(y\mid X)}$, summing over $y\in\cY$, and using
$\int_{\cY}\bE{\p_z}{W(y\mid X)} \dd{\mu}=1$ gives 
\[
\int_{\cY} \frac{\Var_{\p_z}[W(y\mid X)]}{\bE{\p_z}{W(y\mid X)}}\dd{\mu} \leq (e^\priv-1)^2 \int_{\cY}\bE{\p_z}{W(y\mid X)} \dd{\mu}=(e^\priv-1)^2,
\]
thus establishing~\eqref{eq:cor-ldp}. For the bound of 
$e^\priv$, observe 
that, for all $y \in \cY$, 
\[
	\Var_{\p_z}[W(y\mid X)] \le 
	\bE{\p_z}{W(y\mid X)^2} 
	\le e^\priv \min_{x \in \cX}   W(y\mid x)  \bE{\p_z}{W(y\mid X)}.
\]
Hence
\[
\int_{\cY} \frac{\Var_{\p_z}[W(y\mid X)]}{\bE{\p_z}{W(y\mid X)}}\dd{\mu} 
\leq e^\priv\int_{\cY}\min_{x \in \cX} W(y\mid x) \dd{\mu} \le e^\priv 
\cdot 
\min_{x \in \cX} \int_{\cY} W(y\mid x) \dd{\mu} = e^\priv. \qedhere
\]
\new{The bound~\eqref{eqn:cor-ldp:subgaussian} 
(under~\cref{assn:subgaussianity}) will follow 
from~\eqref{eqn:subgaussian-bound}, and the relation between differential privacy and KL divergence. Indeed, the mutual information  $\mutualinfo{\p_z}{W}$ can be rewritten as the expected (over $X\sim \p_Z$) KL divergence between the distribution $\p^{W,X}\eqdef W(\cdot\mid X)$ over $\cY$ induced by the channel $W$ on input $X$, and the distribution $\p_Z^W\eqdef\bE{X'\sim\p_z}{W(\cdot\mid X')}$  over $\cY$ induced by the input distribution $\p_z$ and the channel $W$:
\[
	\mutualinfo{\p_z}{W} = \bE{X\sim\p_z}{ \kldiv{\p^{W,X}}{\p_z^W} }
	= \bE{X\sim\p_z}{ \bE{Y\sim\p^{W,X}}{\ln\frac{W(Y\mid X)}{\bE{X'\sim\p_z}{W(Y\mid X')}}} } \,;
\]
but the $\priv$-LDP condition from~\cref{eq:ldp} guarantees that the log-likelihood ratio in the inner expectation is (almost surely) at most $\priv$, so that $\mutualinfo{\p_z}{W} \leq \priv$ for every $z$ and $W\in\cWpriv[\priv]$. This yields~\eqref{eqn:cor-ldp:subgaussian}.
}
\end{proof}

\begin{corollary}[Communication constraints]
\label{cor:simple-numbits}
For any family of channels $\cW$ with finite output space $\cY$ and any family of distributions $\{\p_z, z\in \bool^\zdims\}$
satisfying~\cref{assn:decomposition-by-coordinates,assn:orthonormal},
with the notation of~\cref{thm:avg:coordinate}, we have
\begin{align}
  \label{eqn:cor-numbits}
 \Paren{\frac{1}{\zdims}\sum_{i=1}^\zdims\totalvardist{\p_{+i}^{Y^\ns}}{\p_{-i}^{Y^\ns}}}^2 
 \le \frac{7}{\zdims} \ns \kappavar  \abs{\cY}.
\end{align}
Moreover, if~\cref{assn:subgaussianity} holds as well, we have
\begin{align}
  \label{eqn:cor-numbits:subgaussian}
 \Paren{\frac{1}{\zdims}\sum_{i=1}^\zdims\totalvardist{\p_{+i}^{Y^\ns}}{\p_{-i}^{Y^\ns}}}^2 
 \le \frac{\new{14}}{\zdims}\ns\sigma^2\log\abs{\cY}.
\end{align}
\end{corollary}
\begin{proof}
In view of~\eqref{eqn:var-bound}, to
establish~\eqref{eqn:cor-numbits}, it suffices to show that
$\frac{\Var_{\p_z}[W(y\mid X)]}{\bE{\p_z}{W(y\mid X)}} \leq 1$ for
every $y\in\cY$. Since $W(y\mid
x)\in(0,1]$ for all $x\in\cX$ and $y\in\cY$, so that
\[
\Var_{\p_z}[W(y\mid X)] \leq \bE{\p_z}{W(y\mid
X)^2} \leq \bE{\p_z}{W(y\mid X)}.
\]
The second bound
(under~\cref{assn:subgaussianity}) will follow
from~\eqref{eqn:subgaussian-bound}. Indeed, recalling that
the entropy of the output of a channel is bounded below by the mutual information between
input and the output,
we have $\mutualinfo{\p_z}{W}\leq H(\p^W_z)$, where $\p_z^W\eqdef\bE{\p_z}{W(\cdot\mid X)}$ is the distribution over $\cY$ induced by the input distribution $\p_z$ and the channel $W$. Using the fact that the  
entropy of a distribution over $\cY$ is at most $\log |\cY|$
in~\eqref{eqn:subgaussian-bound}
gives~\eqref{eqn:cor-numbits:subgaussian}.
\end{proof}
\begin{proofof}{\cref{thm:avg:coordinate}}
Our starting point is~\cref{eqn:main-bound} which holds under~\cref{assn:decomposition-by-coordinates}. We will bound the right-hand-side of~\cref{eqn:main-bound} under assumptions of orthogonality and subgaussianity to prove the two bounds in~\cref{thm:avg:coordinate}.

First, under orthogonality (\cref{assn:orthonormal}), we apply Bessel's inequality to~\cref{eqn:main-bound}. For a fixed $z\in\cZ$, \new{write $\psi_{z,i} = \frac{\phi_{z,i}}{\sqrt{\bE{\p_{z}}{\phi_{z,i}^2}}}$}, and complete $(1,\psi_{z,1},\dots,\psi_{z,\zdims})$ to get an orthonormal basis $\mathcal{B}$ for $L^2(\cX,\p_z)$. Fix any $W\in\cW$ and $y\in\cY$, and, for brevity,
define $a\colon\cX\to\R$ as $a(x)=W(y\mid x)$.  Then, we have
\begin{align*}
    \sum_{i=1}^\zdims \bE{}{\phi_{z,i}(X)a(X)}^2
    &\leq \kappavar\sum_{i=1}^\zdims \bE{}{\psi_{z,i}(X)a(X)}^2 
    = \kappavar\sum_{i=1}^\zdims \dotprod{a}{\psi_{z,i}}^2
= \kappavar\sum_{i=1}^\zdims \dotprod{a-\bEE{a}}{\psi_{z,i}}^2 \\
    &\leq \kappavar\sum_{\psi\in\mathcal{B}} \dotprod{a-\bEE{a}}{\psi}^2= \kappavar\Var[ a(X) ],
\end{align*}
where for the second identity we used the assumption that $\dotprod{\bEE{a}}{\psi_{z,i}}=0$ for all $i\in[\zdims]$ (since $1$ and $\psi_{z,i}$ are orthogonal). This establishes~\cref{eqn:var-bound}.

Turning to~\cref{eqn:subgaussian-bound}, suppose that~\cref{assn:subgaussianity} holds. Fix $z\in\cZ$, and consider any $W\in\cW$ and $y\in\cY$.
Upon applying 
\ifnum\issupplement=1
	\cref{l:basic_mc} 
\else
	Lemma~4 of the Supplement (See Supplement (Appendix~B) for the precise statement and proof) 
\fi
to the \new{$\sigma^2$-subgaussian random 
vector 
$\phi_z(X)$} and with $a(x)$ set to $W(y\mid x)\in [0,1]$, 
we get that
\begin{align*}
    \sum_{i=1}^\zdims \bE{\p_z}{\phi_{z,i}(X)W(y\mid X)}^2
    &= \normtwo{\bE{\p_z}{\phi_z(X)W(y\mid X)}}^2
\\
&\leq 2\sigma^2\bE{\p_z}{W(y\mid X)}
\cdot \bE{\p_z}{W(y\mid X)\log\frac{W(y\mid X)}{\bE{\p_z}{W(y\mid X)}}}
\end{align*}
Integrating over $y\in\cY$, this gives
\begin{align*}
\int_{\cY} \frac{\sum_{i=1}^{\zdims} \bE{\p_{z}}{\phi_{z,i}(X)W(y\mid X)}^2}{\bE{\p_{z}}{W(y\mid X)}} \dd{\mu}
&\leq 2\sigma^2\cdot \int_{\cY} \bE{\p_{z}}{W(y\mid X)\log\frac{W(y\mid X)}{\bE{\p_{z}}{W(y\mid X)}}} \dd{\mu}
\\
&= 2\sigma^2 \mutualinfo{\p_z}{W},
\end{align*}
 which yields the claimed bound.
\end{proofof}

\section{Relation to other lower bound methods} \label{sec:relation}
We now discuss how our techniques compare with other existing approaches for proving lower bounds under information constraints. Specifically, we clarify the relationship between our technique and the approach using strong data processing inequalities (SDPI) as well as that based on van Trees inequality
(a generalization of the Cram\'er--Rao bound).

\subsection{Strong data processing inequalities}
We note first that the bound in~\cref{eqn:subgaussian-bound} can be interpreted as a
strong data processing inequality. Indeed, the average discrepancy on the left-side of inequality can be viewed as the average information $Y^n$ reveals about each bit of $Z$. Here the information is measured in terms of total variation distance. The information quantity on the right-side denotes the information between the input $X^\ns$ and the output $Y^\ns$ of the channels. Since the Markov relation
$Z^\ns \text{ --- } X^\ns\text{ --- }Y^\ns$ holds, the inequality is thus a strong data processing inequality with strong data processing constant roughly \new{$\sigma^2/\zdims$}. Such strong data
processing inequalities were used to derive lower bounds for
statistical estimation under communication constraints
in~\cite{ZDJW:13,BGMNW:16,XR:18}. We note that our approach recovers
these bounds, and further applies to arbitrary constraints captured by
$\cW$. 

\subsection{Connection to the van Trees inequality}
The average information bound in~\eqref{eqn:main-bound}, in fact,
allows us to recover bounds similar to the 
van Trees inequality-based bounds
developed in~\cite{BHO:19} and~\cite{BCO:20}.

For $\Theta\subset\R^\zdims$
and a parametric family\footnote{
We assume that each distribution $\p_\theta$
has a density with respect to a common measure $\nu$,
and, with a slight abuse of notation,
denote the density of $\p_\theta$ also by $\p_\theta(X)$.}
$\cP_\Theta=\{\p_\theta, \theta\in \Theta\}$,
recall
that the Fisher information matrix $J(\theta)$
is a $\zdims \times \zdims$ matrix
 given by, under some mild regularity conditions,
\[
J(\theta)_{i,j}= - \bE{\p_\theta}{
\frac {\partial^2 \log \p_\theta}{\partial \theta_i \partial \theta_j} (X)},
\quad i,j\in [\zdims].
\]
In particular,  the diagonal
entries equal
\[
J(\theta)_{i,i}=
\bE{\p_\theta}{ \left(
\frac 1 {\p_{\theta}(X)}\cdot\frac {\partial \p_\theta}{\partial \theta_i } (X)\right)^2},\quad i\in[\zdims].
\]
For our application, given a channel $W\in \cW$, we consider the family 
$\cP_\Theta^W \eqdef \{\p_\theta^W, \theta\in \Theta\}$ of distributions induced
on the output of the channel $W$ when the input distributions
are from $\cP_\Theta$. We denote the Fisher information matrix for this family by
$J^W(\theta)$, which we compute next 
under a refined version of our~\cref{assn:decomposition-by-coordinates} described below.

Let $\theta$ be a point in the interior of $\Theta$ and $\p_\theta$ be differentiable at $\theta$. We set $\theta_z\eqdef\theta + \frac{\gamma}{2} z$, $z\in \bool^\zdims$, and make the following assumption about the structure
of the parametric family of distribution: For all $z\in \bool^\zdims$ and $i\in [\zdims]$,
\begin{align}
 \dv{\p_{z^{\oplus i}}}{\p_z}=1+\gamma \xi_{z,i}^\gamma + \gamma^2\psi_{z,i}^\gamma,
 \label{eqn:refined_assumption}
 \end{align}
where $\bE{\p_z}{\xi_{z,i}^\gamma(X)^2}$ and $\bE{\p_z}{\psi_{z,i}^\gamma(X)^2}$
are assumed to be uniformly bounded for $\gamma$ sufficiently small;
for concreteness, we assume
$\bE{\p_z}{\psi_{z,i}^\gamma(X)^2}\leq c^2$ for a constant $c$, for all $\gamma$ sufficiently small.
Let $\xi_{z,i}(x)\eqdef \lim_{\gamma\to \zero}
\xi_{z,i}^\gamma(x)$, for all $x$.

In applications, we expect the dependence of $\xi_{z,i}^\gamma$ on $\gamma$ to be ``mild,''
and, in essence,  
the assumption above provides a linear expansion of the term $\alpha_{z,i}\phi_{z,i}$ from~\cref{assn:decomposition-by-coordinates} as a function of the perturbation parameter $\gamma$.
Assuming that the densities are differentiable as a function of $\theta$,
for the distribution $\p_\theta^W$ of the output of a channel $W$ with input $X\sim \p_\theta$,
 we get
\begin{align*}
\frac{\partial \p_\theta^{W}(y)}{\partial \theta_i}
&= z_i\lim_{\gamma\to 0} \frac{\p_{\theta_z}^{W}(y) - \p_{\theta_{z^{\oplus i}}}^{W}(y)}
{\gamma}
\\
&= z_i \lim_{\gamma\to 0}\bE{\p_z}{(\xi_{z,i}^\gamma(X)+ \gamma\psi_{z,i}^\gamma(X))W(y\mid X)}
\\
&=z_i \bE{\p_\theta}{\xi_{z,i}W(y\mid X)},
\end{align*}
where we used~\cref{eqn:refined_assumption},
the fact that $\lim_{\gamma\to 0}\theta_z = \theta$,
the fact that
$\bE{\p_z}{\psi_{z,i}^\gamma(X)W(y\mid X)}\leq c\sqrt{\bE{\p_z}{W(y\mid X)^2}}\leq c$, and
the dominated convergence theorem. Thus, we get
\begin{align}
\Tr(J^W(\theta)) =\sum_{i=1}^\zdims \int_{\cY} \frac{\bE{\p_{\theta}}{\xi_{z,i}(X)W(y\mid X)}^2}{\bE{\p_{\theta}}{W(y\mid X)}} \dd{\mu}.
\label{e:trace_Fisher}
\end{align}

Our information contraction bound will be seen later (\cref{sec:applications}) to yield lower bounds for expected estimation error.
For concreteness, we give a preview of a version here.
We assume for simplicity that $\cW_t=\cW$ for all $t$ and consider the $\lp[2]$ loss function
for the dense ($\sparam = 1/2$) case.  
By following the proof of~\cref{lem:mean:estimation} below, given an
$(\ns, \dst)$-estimator $\hat{\theta}=\hat{\theta}(Y^\ns,U)$ of $\cP_\Theta$ using $\cW^\ns$
under $\lp[2]$ loss, we can find an estimator
$\hat{Z}=\hat{Z}(Y^\ns, U)$ such that
\begin{align*}
\gamma^2\sum_{i=1}^{\zdims}\bPr{\hat{Z}_i \neq Z_i} =  \bEE{\norm{\theta_{\nohat{Z}}-\theta_{\hat{Z}}}_2^2}\leq 4\dst^2,
\end{align*}
whereby
\begin{align*}
\frac{1}{\zdims}\sum_{i=1}^\zdims\totalvardist{\p_{+i}^{Y^\ns}}{\p_{-i}^{Y^\ns}} 
&\geq 1-\frac{2}{\zdims}\sum_{i=1}^\zdims\bPr{\hat{Z}_i \neq Z_i} 
\geq 1- \frac {8\dst^2} {\zdims\gamma^2}.
\end{align*}
Upon setting $\gamma\eqdef 4\dst/\sqrt{\zdims}$,
we get that the left-side of~\cref{eqn:main-bound} is bounded
below by $1/4$. 
For the same $\gamma$ and under~\cref{eqn:refined_assumption},
the right-side evaluates to
\begin{align*} 
\lefteqn{\frac{4\gamma^2\ns}{\zdims}  \max_{z\in\cZ}\max_{W\in\cW} \sum_{i=1}^\zdims \int_{\cY} \frac{\bE{\p_{z}}{(\xi_{z,i}^\gamma(X)+\gamma\psi_{z,i}^\gamma(X))W(y\mid X)}^2}{\bE{\p_{z}}{W(y\mid X)}} \dd{\mu}}
\\
&\leq {\frac{8\gamma^2\ns}{\zdims}  \max_{z\in\cZ}\max_{W\in\cW} \sum_{i=1}^\zdims \int_{\cY} \frac{\bE{\p_{z}}{\xi_{z,i}^\gamma(X)W(y\mid X)}^2
+\gamma^2\bE{\p_{z}}{\psi_{z,i}^\gamma(X)W(y\mid X)}^2
}{\bE{\p_{z}}{W(y\mid X)}} \dd{\mu}}
\\
&\leq \frac{128 \dst^2 \ns}{\zdims^2}\left(
\max_{z\in\cZ}\max_{W\in\cW} \sum_{i=1}^\zdims
\int_{\cY} \frac{\bE{\p_{z}}{\xi_{z,i}^\gamma(X)W(y\mid X)}^2}{\bE{\p_{z}}{W(y\mid X)}} \dd{\mu}
+  c^2 \dst^2\right),
\end{align*}
where we used $(a+b)^2\leq 2(a^2+b^2)$ and
\[
\int_{\cY}
\frac{\bE{\p_z}{\psi_{z,i}^\gamma(X)W(y\mid X)}^2}{\bE{\p_z}{W(y\mid X)}}
\dd{\mu}
\leq
\int_{\cY}\bE{\p_z}{\psi_{z,i}^\gamma(X)^2W(y\mid X)}\dd{\mu}=\bE{\p_z}{\psi_{z,i}^\gamma(X)^2}
\leq c^2.
\]
Therefore,~\cref{eqn:main-bound} yields
\[
\dst^2\geq 
\frac{\zdims^2}{
256\cdot
\ns\,\left(\max_{z\in\cZ}\max_{W\in\cW}\sum_{i=1}^\zdims  \int_{\cY} \frac{\bE{\p_{z}}{\xi_{z,i}^\gamma(X)W(y\mid X)}^2}{\bE{\p_{z}}{W(y\mid X)}} \dd{\mu} + 
 c^2\right)}.
\]
This bound is, in effect, the same as the van Trees inequality with $\Tr(J^W(\theta))$ replaced by
\[
  g(\gamma)\eqdef \sum_{i=1}^\zdims  \int_{\cY} \frac{\bE{\p_{z}}{\phi_{z,i}(X)W(y\mid X)}^2}{\bE{\p_{z}}{W(y\mid X)}} \dd{\mu}.
\]
In fact, in view of~\cref{e:trace_Fisher}, $\Tr(J^W(\theta))=\lim_{\gamma\to 0}g(\gamma) =: g(0)$. 
Thus, our general lower bound will recover van Trees inequality-based bounds when~\cref{eqn:refined_assumption} holds and 
$g(\gamma)\approx g(0)$. We note that~\cref{eqn:refined_assumption} holds for all the families considered in this paper (see~\cref{eqn:prod-bernoulli_assn} for product Bernoulli,~\cref{eqn:gaussian_assn}
for Gaussian, and~\cref{eqn:discrete_assn} for discrete distributions).
We close this discussion by noting that results in~\cref{ss:plug-play}
are obtained by  deriving bounds for $g(\gamma)$ which apply for all $\gamma$ and, therefore,
also for $g(0)=\Tr(J^W(\theta))$.
\hnote{Do we need to comment on how $g(\gamma)$ and $g(0)$ compare for our
families of interest, when $\gamma\approx \dst/\sqrt{\zdims}$? That will complete
the comparison, but it will take some calculations.}
\cnote{I think it's fine not to; if a reviewer asks for it, then we'll cross that bridge.}

%% file: sec-assouad.tex
\newcommand{\cst}{{\color{red}\clubsuit}}

In the previous section we provided an upper bound on 
 $\frac{1}{\zdims}\sum_{i=1}^\zdims\totalvardist{\p_{+i}^{Y^\ns}}{\p_{-i}^{Y^\ns}}$.
  We now prove a lower bound for this quantity in terms of the parameter 
 estimation task we set out to solve. This is an ``Assouad's lemma-type'' 
 bound, which when combined with~\cref{thm:avg:coordinate} will establish 
 the bounds for $\ns$; and, reorganizing, the minimax rate lower bounds. 

To state the result, we require the following assumption, which \new{relates} the $\lp[\prm]$ distance between the parameters $\theta_z$s \new{to} the corresponding distance between $z$s. 
\begin{assumption}[Additive loss]
  \label{assn:additive-loss}
Fix $\prm\in[1,\infty)$. For every  
$z,z^\prime\in\cZ\subset\bool^\zdims$,
\begin{align*}
\lp[\prm](\theta_z,\theta_{z^\prime}) = 4\dst \Paren{\frac{\dist{z}{z^\prime}}{\sparam\zdims}}^{1/p},
\end{align*}
where $\dist{z}{z^\prime} \eqdef \sum_{i=1}^\zdims\indic{z_i\neq z_i^\prime}$ denotes the Hamming distance.\footnote{\new{Note that this assumption can be relaxed to \smash{$\lp[\prm](\theta_z,\theta_{z^\prime}) \asymp \dst \Paren{\frac{\dist{z}{z^\prime}}{\sparam\zdims}}^{1/p}$} (equality up to constants), as the upper bound is only required for the proof of~\cref{lem:mean:estimation}; for the sake of exposition, and as it suffices for all our applications, we here state the simpler version.}}
\end{assumption}

\begin{lemma}[Assouad-type bound]
	\label{lem:mean:estimation}
	Let $\prm\ge 1$ and assume that $\{\p_z, z\in \cZ\}$, 
	$\sparam\in[0,1/2]$ satisfy~\cref{assn:additive-loss}. Let $Z$ be a 
	random variable on $\cZ=\bool^\zdims$ with distribution 
	$\rademacher{\sparam}^{\otimes\zdims}$.
	Suppose that $(\Pi, \hat{\theta})$ constitutes an $(\ns, \dst)$-estimator 
	of $\cP_\Theta$
	using $\cW^\ns$
	under $\lp[\prm]$ loss (see~\cref{def:estimate}) and $\bP{Z}{\p_Z 
	\in 
	\cP_\Theta} \ge 1 -  \sparam/4$.
	Then the 
	transcript $(Y^\ns, U)$ of $\Pi$ satisfies
	\begin{align*}
	\frac{1}{\zdims}\sum_{i=1}^\zdims\totalvardist{\p_{+i}^{Y^\ns}}{\p_{-i}^{Y^\ns}}
	 \geq \frac{1}{4},
	\end{align*}
	where $\p_{+i}^{Y^\ns}\eqdef \bEEC{\p_Z^{Y^\ns}}{Z_i= + 1}$, 
	$\p_{-i}^{Y^\ns}\eqdef \bEEC{\p_Z^{Y^\ns}}{Z_i= - 1}$.
\end{lemma}
\begin{proof}
	Given an $(\ns, \dst)$-estimator $(\Pi, \hat{\theta})$, define an estimate 
	$\hat{Z}$ for $Z$ as 
	\[
	\hat{Z} \eqdef \underset{z\in\cZ}{\arg\!\min} 
	\norm{\theta_z-\hat{\theta}(Y^\ns, U)}_\prm.
	\]
	By the triangle inequality,
	\[
	\norm{\theta_{\nohat{Z}}-\theta_{\hat{Z}}}_\prm\leq \norm{
		\theta_{\nohat{Z}}-\hat{\theta}(Y^\ns, U)}_\prm+
	\norm{\theta_{\hat{Z}}-\hat{\theta}(Y^\ns, U)}_\prm
	\leq 2\norm{\hat{\theta}(Y^\ns, U)-\theta_Z}_\prm.
	\]
	
	Since $(\Pi, \hat{\theta})$ is an $(\ns,\dst)$-estimator
	under $\lp[\prm]$ loss for $\cP_\Theta$,
	\begin{align}
	\bE{Z}{\bE{\p_Z}{\norm{\theta_{\nohat{Z}}-\theta_{\hat{Z}}}_\prm^\prm}}
	& \leq
	2^\prm\dst^\prm \bPr{\p_Z \in \cP_\Theta} + \max_{z \neq z'}  
	\nonumber 
	\norm{\theta_{z}-\theta_{z'}}_\prm^\prm \bPr{\p_Z \notin \cP_\Theta}  
	\\
	& \le  2^\prm\dst^\prm + 4^\prm\dst^\prm  \frac{1}{\sparam} \cdot 
	\frac{\sparam}{4} \label{eqn:total-error}\\
	& \le \frac{3}{4} 4^\prm\dst^\prm,
	\label{eqn:theta-Z}
	\end{align}
	\noindent where \cref{eqn:total-error} follows from 
	\cref{assn:additive-loss} and $\bPr{\p_Z \in 
		\cP_\Theta} \ge 1 -  \sparam/4$. Next, for $\prm \in [1, \infty)$, 
		by~\cref{assn:additive-loss}, 
	$\norm{\theta_{\nohat{Z}}-\theta_{\hat{Z}}}_\prm^\prm \geq 
	\frac{4^\prm\dst^\prm}{\sparam\zdims}\sum_{i=1}^\zdims 
	\indic{Z_i\neq \hat{Z}_i}$. Combining with~\cref{eqn:theta-Z} this shows 
	that 
	$\frac{1}{\sparam\zdims}\sum_{i=1}^\zdims \probaOf{Z_i\neq \hat{Z}_i} 
	\leq \frac{3}{4} \,.$\cmargin{The $\infty$ norm 
	argument does not play well with sparsity at all. I suggest we just do the 
	naive thing and use $\prm = \log\zdims$ to derive our $\lp[\infty]$ 
	lower bounds.}
	
	Furthermore, since the Markov relation $Z_i-(Y^\ns,U)-\hat{Z}_i$ holds 
	for all $i$, we can lower bound $\probaOf{Z_i\neq \hat{Z}_i}$ using the 
	standard relation between total variation distance and hypothesis testing 
	as follows, using that $\tau \leq 1/2$ in the second inequality:
	\begin{align*}
	\probaOf{Z_i\neq \hat{Z}_i} 
	&\geq \sparam \probaCond{ \hat{Z}_i = -1 }{Z_i = 1} + (1-\sparam) 
	\probaCond{ \hat{Z}_i = 1 }{Z_i = -1} \\
	&\geq \sparam \Paren{\probaCond{ \hat{Z}_i = -1 }{Z_i = 1} + 
	\probaCond{ \hat{Z}_i = 1 }{Z_i = -1} } \\
	&\geq \sparam \Paren{1- \totalvardist{\p_{+i}^{Y^\ns}}{\p_{-i}^{Y^\ns}} 
	}\,.
	\end{align*}
	Summing over $1\leq i\leq \zdims$ and combining it with the previous 
	bound, we obtain
	\[
	\frac{3}{4} \geq \frac{1}{\sparam\zdims}\sum_{i=1}^\zdims 
	\probaOf{Z_i\neq \hat{Z}_i} \geq 1 - 
	\frac{1}{\zdims}\sum_{i=1}^\zdims\totalvardist{\p_{+i}^{Y^\ns}}{\p_{-i}^{Y^\ns}}
	\]
	and reorganizing proves the result.
\end{proof}

%% file: sec-applications.tex
We now consider three distribution families: product Bernoulli distributions and Gaussian distributions with identity covariance matrix (and $s$-sparse mean vectors), and discrete distributions (multinomials), to illustrate the generality and efficacy of our bounds. We describe these three families below, before addressing each of them in their respective subsection. 
\begin{description}
	\item [Sparse Product Bernoulli distributions ($\mathcal{B}_{\dims,s}$).]  Let $1\leq s\leq \dims$, $\Theta=\setOfSuchThat{\theta\in[-1,1]^\dims}{\norm{\theta}_0 \leq s}$, and $\cX=\{-1,1\}^\dims$. Let $\cP_\Theta\eqdef \mathcal{B}_{\dims,s}$ be the family of $\dims$-dimensional $s$-sparse product Bernoulli distributions over $\cX$. Namely, for $\theta=(\theta_1,\ldots,\theta_\dims)\in\Theta$, the distribution $\p_\theta$ is equal to $\otimes_{j=1}^\dims\rademacher{\frac{1}{2}(\theta_j+1)}$: a distribution on $\bool^\dims$ such that the marginal distributions are independent, and for which the mean of the $j$th marginal is $\theta_j$.
\item [\new{Sparse} Gaussian distributions ($\mathcal{G}_{\dims,s}$).] 
Let $1\leq s\leq \dims$, $\Theta=\setOfSuchThat{\theta\in[-1,1]^\dims}{\norm{\theta}_0 \leq s}$, and $\cX=\R^\dims$.
Let $\cP_\Theta\eqdef \mathcal{G}_{\dims,s}$ be the family of $\dims$-dimensional spherical Gaussian distributions with bounded $s$-sparse mean. That is, for $\theta\in\Theta$, $\p_\theta=\gaussian{\theta}{\II}$ with mean $\theta$ and covariance matrix $\II$. We note that this general formulation assumes $\norminf{\theta}\leq 1$ (from the choice of $\Theta$).\footnote{This assumption that the mean is bounded, in our case in $\lp[\infty](\textbf{0},1)$, is standard, and necessary in order to obtain finite upper bounds: indeed, a packing argument shows that if the mean is assumed to be in a ball of radius $R$, then a $\log^{\Omega(1)} R$ dependence in the sample complexity is necessary in both the communication-constrained and LDP settings. Our choice of radius $1$ is arbitrary, and our upper bounds can be generalized to any $R\geq 1$.}
\item [Discrete distributions ($\distribs{\dims}$).] Let $\Theta=\setOfSuchThat{ \theta\in[0,1]^{\dims} }{ \sum_{i=1}^{\dims}\theta_i = 1 }\subseteq \R^{\dims}$ and $\cX=\{1,2,\dots,\dims\}$.
Let $\cP_\Theta \eqdef \distribs{\dims}$, where $\distribs{\dims}$ is the standard $(\dims-1)$-simplex of all probability mass functions over $\cX$. Namely, 
the distribution $\p_\theta$ is a distribution on $\cX$, where, for $j\in[\dims]$, the probability assigned to the element $j$ is $\p_\theta(j) = \theta_{j}$.
For a unified presentation, we view $\theta$ as the mean vector of the
``categorical distribution,'' namely the distribution of vector $(\indic{X=x}, x\in \cX)$
for $X$ with distribution $\p_\theta$. 
\end{description}
\noindent We now define our measure of interest, the minimax error rate
of mean estimation, and then instantiate it for $\mathcal{B}_{\dims,s}$, $\mathcal{G}_{\dims,s}$, and $\distribs{\dims}$. 
\begin{definition}[Minimax rate of mean estimation]
\label{def:error-rate-mean}
Let $\cP$ be a family of distributions parameterized by $\Theta\subseteq\R^\dims$. For $\prm\in[1,\infty]$, $\ns\in\N$, and a family of channels $\cW$, the minimax error rate of mean estimation for $\cP$ using $\cW^\ns$ under $\lp[\prm]$ loss, denoted
$\cE_{\prm}(\cP, \cW, \ns)$, is the least $\dst\in (0,1]$
such that there exists an $(\ns, \dst)$-estimator
for $\cP$ using $\cW$ under $\lp[\prm]$ loss (see~\cref{def:estimate}).
\end{definition}

We obtain lower bounds on the minimax rate of mean estimation for the different families
above by specializing our general bound. 

We remark that our methodology is not specific to $\lp[\prm]$ losses, and can be used for arbitrary additive losses such as (squared) Hellinger or, indeed, for any loss function for which an analogue of~\cref{lem:mean:estimation} can be derived. Also, as another remark, we note that our lower bounds for $\lp[\prm]$ losses are not implied by those for, say, $\lp[1]$ or $\lp[2]$ losses. Indeed, while by H\"older's inequality a rate lower bound for,~\eg{} mean estimation under the $\lp[2]$ loss implies a similar lower bound under $\lp[\prm]$ loss for any $\prm \geq 2$, those impose a strong restriction on the range of $\ns$, typically $\ns \geq \dims^c$ for some $c>0$. This in turn severely restricts the applicability of the result, which only holds asymptotically for vanishing rate. Our lower bounds under $\lp[\prm]$ loss do not suffer such restrictions.  \todonote{Check this paragraph.}

\cmargin{The table was underwhelming and confusing (hard to present things in an interesting way, and this actually made our results look bad since we could only put a subset, and not the most interesting across the board. Removed.}

\subsection{Product Bernoulli family}
\label{sec:bernoulli-product}
\input{sec-bernoulli}
\subsection{Gaussian family}
\label{sec:gaussian}
\input{sec-gaussian-new}
\subsection{Discrete distribution estimation}
\label{sec:discrete}
\input{sec-discrete}

%% file: sec-bernoulli.tex
For consistency with mean estimation problems, we denote the parameter $\theta$ associated with $\p\in \mathcal{B}_{\dims,s}$ by $\mu(\p)$ and the corresponding estimator by $\hat{\mu}$
instead of $\hat{\theta}$. %

 We will establish the following sample complexity bounds for estimating $\mathcal{B}_{\dims,s}$ under privacy and communication constraints.\footnote{For simplicity, we focus on the case where the communication (resp., privacy) parameters are the same for all players; but our lower bounds immediately generalize to different $\numbits_t$'s (resp, $\priv_t$) for each player; see~\cref{rk:differentconstraints}.} We focus in this section on establishing the lower bounds, and defer the upper bounds to 
 \ifnum\issupplement=1
	\cref{app:ub:bernoulli}.
\else
	the Supplement (Appendix~C.1); 
\fi
\begin{theorem}
\label{theorem:mean:estimation:bernoulli}
Fix $\prm\in[1,\infty)$. For $4\log \dims\leq s\leq \dims$, 
\new{$\priv\in(0,\infty)$}, and $\numbits \geq 1$, 
\begin{equation} \label{eqn:bern-est:ldp}
 \sqrt{ \frac{\dims s^{2/\prm}}{\ns\new{(\priv^2\land\priv)}} \newchange{\lor\frac{s^{2/\prm}\log\frac{2\dims}{s}}{\ns} } } \land 1 
 \lesssim \cE_{\prm}(\mathcal{B}_{\dims,s} , \cWpriv[\priv], \ns) \lesssim 
 \sqrt{ \frac{\dims s^{2/\prm}}{\ns\new{(\priv^2\land \newchange{\priv})}} 
 \newchange{\lor\frac{s^{2/\prm}\log\frac{2\dims}{s}}{\ns} }  }
\end{equation}
and
\begin{equation} \label{eqn:bern-est:comm}
 \sqrt{\frac{\dims s^{2/\prm}}{\ns\numbits}\lor\frac{s^{2/\prm}\log\frac{2\dims}{s}}{\ns}} \land 1 \lesssim \cE_{\prm}(\mathcal{B}_{\dims,s} , \cWcomm[\numbits], \ns) \lesssim \sqrt{\frac{\dims s^{2/\prm}}{\ns\numbits}\lor\frac{s^{2/\prm}\log\frac{2\dims}{s}}{\ns}}
\end{equation}
 For  $\prm=\infty$, we have the upper bounds
\begin{align*}
\cE_{\infty}(\mathcal{B}_{\dims,s} , \cWpriv[\priv], \ns) = \bigO{\sqrt{\frac{\dims\log s}{\ns\priv^2}}} \quad \text{and} \quad 	\cE_{\infty}(\mathcal{B}_{\dims,s} , \cWcomm[\numbits], \ns)=\bigO{\sqrt{\frac{\dims\log s}{\ns\numbits}\lor \frac{\log \dims}{\ns}} },
\end{align*}
while the lower bounds given in~\cref{eqn:bern-est:ldp,eqn:bern-est:comm} hold for $\prm=\infty$, too.\footnote{That is, the upper and lower bounds only differ by a $\log s$ factor for $\prm=\infty$.}
\end{theorem}
\begin{remark}\label{rk:interactive:gap}
  It is worth noting that previous work had shown, in the simpler \emph{noninteractive} model, a rate lower bound scaling as $\sqrt{\dims s/(\ns\numbits) \log(2\dims/s)}$  for the specific case of $\lp[2]$ loss (see, for instance,~\cite[Theorem 7]{HOW:18:v3} for the sparse Gaussian case, which implies the Bernoulli one). An analogous phenomenon was observed for local privacy (\eg~\cite{DJW:17}). Thus, by removing this logarithmic factor from the upper bound, our result establishes the first (to the best of our knowledge) separation between interactive and noninteractive protocols for sparse mean estimation under communication or local privacy constraints. A similar observation holds for~\cref{theorem:mean:estimation:gaussian} below, for sparse Gaussian mean estimation under local privacy.
\end{remark}

\begin{proof}
Fix $\prm\in[1,\infty)$. Let $\zdims=\dims$, $\cZ=\bool^\dims$, and 
$\sparam=\frac{s}{2\dims}$; and suppose that, for some $\dst \in 
(0,1/8]$, 
there exists an $(\ns,\dst)$-estimator for $\mathcal{B}_{\dims,s}$ under 
$\lp[\prm]$ loss.
We fix a parameter $\gamma\in(0,1/2]$, which will be chosen as a function 
of $\dst,\dims,\prm$ later. Consider the set of $2^\dims$ product 
Bernoulli distributions $\{\p_z\}_{z\in\cZ}$, where $\mu(\p_z)=\mu_z 
\eqdef \frac{1}{2}\gamma (z+\one[\dims])$
\new{(so the sparsity of the mean vector is equal to the number of positive coordinates of $z$)}. 
We have, for $z\in\cZ$, 
\begin{equation}
\nonumber
  \p_z(x) = \frac{1}{2^\dims}\prod_{i=1}^\dims \Paren{1+\frac{1}{2}\gamma (z_i+1)x_i}, \qquad x\in\cX.
\end{equation} 
It follows for $z\in\cZ$ and $i\in[\dims]$ that 
 \begin{align}
    \p_{z^{\oplus i}}(x) =  \frac{1+\frac{1}{2}\gamma(1-z_i)x_i}{1+\frac{1}{2}\gamma(1+z_i)x_i}
    \p_{z}(x)
      = \Paren{1-\gamma\frac{z_ix_i}{1+\frac{1}{2}\gamma(1+z_i)x_i}}\p_{z}(x)
= \Paren{1+\phi_{z,i}(x)}\p_{z}(x) \label{eqn:prod-bernoulli_assn}
 \end{align}
 where $\phi_{z,i}(x)\eqdef -\frac{\gamma z_ix_i }{1+\frac{1}{2}\gamma(1+z_i)x_i}$.
We can verify that, for $i\neq j$,
 \[
    \bE{\p_z}{\phi_{z,i}(X)} = 0,\quad \bE{\p_z}{\phi_{z,i}(X)^2} = \frac{\gamma^2}{1-\frac{1}{2}\gamma^2(1+z_i)}, \text{ and } \bE{\p_z}{\phi_{z,i}(X)\phi_{z,j}(X)} = 0, 
 \]
 so that~\cref{assn:decomposition-by-coordinates,assn:orthonormal} are satisfied for $\kappavar\eqdef 2\gamma^2$.
Moreover, using, \eg Hoeffding's lemma ($cf.$~\cite{Boucheron:13}),
for $\gamma < 1$, the random vector $\phi_{z}(X) = (\phi_{z,i}(X))_{i\in[\dims]}$ is \new{$\frac{\gamma^2}{(1-\gamma^2)^2}$}-subgaussian. \new{Thus,}~\cref{assn:subgaussianity} holds as well, and we can invoke both parts of~\cref{thm:avg:coordinate}.
 
 	Let $\numone{z} \eqdef \abs{\{ i \in [\dims] \mid z_i = 1\}}$, so that 
 	$\norm{\mu_z}_0 = \sum_{i=1}^\dims \frac{1}{2}(1+z_i) = \numone{z} 
 	$. 
 	The next 
 	claim, which follows from standard bounds for binomial random 
 	variables, 
 	states 
 	that when $Z 
 	\sim 
 	\rademacher{\sparam}^{\otimes\dims}$,  $\mu_Z$ is $s$-sparse with 
 	high probability.
 	\begin{fact} \label{prop:numone}
 		Let $Z 
 		\sim 
 		\rademacher{\sparam}^{\otimes\dims}$, where $\sparam\dims \ge  4\log 
 		\dims$. Then
 		$
 		\bPr{\numone{Z} \leq  2\sparam\dims} \geq 1- \tau/4.
 		$
 	\end{fact}
 	\zmargin{Do we need a proof for this?}
 	\cmargin{Naaay.}
 	\noindent Hence the construction satisfies $\bP{Z}{\p_Z \in 
 		\mathcal{B}_{\dims, s}} \le 1 - \tau/4$, as required in 
 	\cref{lem:mean:estimation}.

We now choose $\gamma =  \gamma(\prm) \eqdef 
\frac{4\dst}{(s/2)^{1/\prm}} \in (0,1/2]$,
which implies that~\cref{assn:additive-loss} holds since
\[
    \lp[\prm](\mu(\p_z),\mu(\p_{z'})) = \gamma\dist{z}{z'}^{1/\prm} = 4\dst\mleft(\frac{\dist{z}{z'}}{\sparam\dims}\mright)^{1/\prm}.
\]
Therefore, we can apply~\cref{lem:mean:estimation} as well.
\new{For $\cWpriv[\priv]$, we prove the two parts of the lower bound separately, depending on whether $\priv \leq 1$. First,} upon combining 
the bounds obtained by~\cref{cor:ldp,lem:mean:estimation} \new{(specifically, for the former,~\eqref{eq:cor-ldp})}, we 
get%
\[
\dims\leq 112\ns\kappavar(e^\priv-1)^2,
\]
whereby, upon recalling that \new{$\kappavar=2\gamma^2$},
and using the value of $\gamma=\gamma(\prm)$ above,
it follows that
\[
\frac{1}{\new{3584}}\cdot\frac{\dims (s/2)^{\frac 2 
\prm}}{\ns(e^\priv-1)^2}\leq \dst^2.
\]
Thus, $\cE_{\prm}(\mathcal{B}_{\dims,s} , \cWpriv[\priv], \ns) = 
\bigOmega{\sqrt{ \frac{\dims s^{2/\prm}}{\ns\priv^2} }}$ for $\priv \in 
(0,1]$. \new{For the second part of the bound, which dominates for $\priv> 1$, observe that~\cref{assn:subgaussianity} holds with \smash{$\sigma^2\eqdef \frac{\gamma^2}{(1-\gamma^2)^2} \leq 2\gamma^2$}; allowing us to apply the second part of~\cref{cor:ldp},~\eqref{eqn:cor-ldp:subgaussian}, which as before combined with~\cref{lem:mean:estimation} yields
\[
\dims\leq 224\ns\sigma^2\priv \leq 448\ns\gamma^2\priv,
\]
and again from the setting of $\gamma$ we get $\cE_{\prm}(\mathcal{B}_{\dims,s} , \cWpriv[\priv], \ns) = 
\bigOmega{\sqrt{ \frac{\dims s^{2/\prm}}{\ns\priv} }}$.
}

Similarly, for $\cWcomm[\numbits]$, again since~\cref{assn:subgaussianity} holds with \new{$\sigma^2 \leq 2\gamma^2$},
upon combining the bounds obtained by~\cref{cor:simple-numbits,lem:mean:estimation}, we get
\[
\frac{\dims s^{\frac{2}{\prm}}}{\new{28672}\ns\numbits} \le \dst^2,
\]
which gives
$\cE_{\prm}(\mathcal{B}_{\dims,s} , \cWcomm[\numbits], \ns) = 
\Omega\big(\sqrt{ \frac{\dims s^{2/\prm}}{\ns\numbits} } \land 1\big)$. 
Finally, note that for $\numbits \geq \frac{\dims}{\log(2\dims/s)}$ 
\newchange{(and similarly for $\priv \geq \frac{\dims}{\log(2\dims/s)}$)}, 
the lower bound follows from the minimax rate in the unconstrained 
setting, which can be seen to be 
$\Omega\big(\sqrt{{s^{2/\prm}\log(2\dims/s)}/{\ns}} 
\big)$~\cite{Wainwright09,Wu20}. This completes the proof.

This handles the case $\prm\in[1,\infty)$. For $\prm=\infty$, the lower 
bounds immediately follow from plugging $\prm = \log s$ in the previous 
expressions, as discussed in~\cref{ft:holder:logd}. 
\end{proof}
\begin{remark}
  \label{rk:differentconstraints}
Although we stated for simplicity the lower bounds of~\cref{theorem:mean:estimation:bernoulli} in the case where all $\ns$ players have a similar local constraints (\ie, same privacy parameter $\priv$, or same bandwidth constraint $\numbits$), it is immediate to check from the application of~\cref{thm:avg:coordinate} that the result extends to different constraints for each player; replacing \new{$\ns(\priv^2\land\priv)$} and $\ns\numbits$ in the statement by \new{$\sum_{t=1}^\ns \priv_t^2\land\priv_t$} and $\sum_{t=1}^\ns \numbits_t$, respectively. A similar remark applies to~\cref{theorem:mean:estimation:gaussian,theorem:mean:estimation:discrete} as well.
\end{remark}

%% file: sec-gaussian-new.tex
Similar to the previous section, we denote the mean by $\mu$ instead of
$\theta$, denote the estimator by $\hat{\mu}$, and consider the minimax 
error rate 
$\cE_{\prm}(\mathcal{G}_{\dims,s}, \cW, \ns)$
of  mean estimation for
$\cP_\Theta=\mathcal{G}_{\dims,s}$
using $\cW$ under $\lp[\prm]$ loss.

We derive a lower bound for $\cE_{\prm}(\mathcal{G}_{\dims,s}, \cW, \ns)$
under local privacy (captured by $\cW=\cWpriv[\priv]$) and communication
(captured by $\cW=\cWcomm[\numbits]$) 
constraints.\footnote{As in the Bernoulli case, we here focus for simplicity 
on the case where the communication (resp., privacy) parameters are the 
same for all players, but our lower bounds easily extend.} Recall that for 
product Bernoulli
mean  estimation
we had optimal bounds for both privacy and communication constraints
for all finite $\prm$. For Gaussians, we will obtain tight bounds for
privacy constraints for $\priv \in (0, 1]$. 
However, for communication
constraints \new{and privacy constraints when $\priv \ge 1$},  our bounds for 
Gaussian distributions are 
tight only in specific regimes of $\ns$ \new{up to logarithmic factors}. We 
state our general
result and provide some remarks before providing the
proofs.

We defer the estimation schemes and their analysis (\ie{} upper
bounds) to 
\ifnum\issupplement=1
	\cref{app:ub:gaussian};
\else
	the Supplement (Appendix~C.2); 
\fi
they follow from a simple 
reduction
from the Gaussian estimation problem to the product Bernoulli one,
which enables us to invoke the protocols for the latter task in both
the communication-constrained and locally private settings.
\begin{theorem}
	\label{theorem:mean:estimation:gaussian}
	Fix $\prm\in[1,\infty)$. \new{For $4\log \dims\leq 
	s\leq \dims$, under LDP constraints, 
	when $\priv\in(0,1]$,  }
	\begin{equation} \label{eqn:gauss-est:ldp}
	\sqrt{ \frac{\dims s^{2/\prm}}{\ns\priv^2} } \land 1 \lesssim 
	\cE_{\prm}(\mathcal{G}_{\dims,s} , \cWpriv[\priv], \ns) \lesssim \sqrt{ 
	\frac{\dims s^{2/\prm}}{\ns\priv^2} }  
	\end{equation}
	and \new{when $\priv > 1$,
		\begin{equation} \label{eqn:gauss-est:ldp-largerho}
		\sqrt{ \frac{\dims s^{2/\prm}}{\ns\priv \log{(\ns\dims)}}\newchange{\lor\frac{s^{2/\prm}\log\frac{2\dims}{s}}{\ns} }  } \land 1 
		\lesssim 
		\cE_{\prm}(\mathcal{G}_{\dims,s} , \cWpriv[\priv], \ns) \lesssim \sqrt{ 
			\frac{\dims s^{2/\prm}}{\ns \newchange{\priv}}   
			\newchange{\lor\frac{s^{2/\prm}\log\frac{2\dims}{s}}{\ns} }  }   
		\end{equation}
}
	\new{Under communication constraints, 
	\begin{equation} \label{eqn:gauss-est:comm}
	\sqrt{\frac{\dims\sprs^{2/\prm}}{\ns\numbits 
	\log(\dims \ns)}\lor\frac{s^{2/\prm}\log\frac{2\dims}{s}}{\ns}}
	\land 1 \lesssim \cE_{\prm}(\mathcal{G}_{\dims,s} , 
	\cWcomm[\numbits], \ns) \lesssim \sqrt{\frac{\dims 
			s^{2/\prm}}{\ns\numbits}\lor\frac{s^{2/\prm}\log\frac{2\dims}{s}}{\ns}}
	\end{equation}
}
	For  $\prm=\infty$, we have the upper bounds
	\begin{align*}
	\cE_{\infty}(\mathcal{G}_{\dims,s} , \cWpriv[\priv], \ns) = 
	\bigO{\sqrt{\frac{\dims\log s}{\ns\priv^2}}} \quad \text{and} \quad 	
	\cE_{\infty}(\mathcal{G}_{\dims,s} , \cWcomm[\numbits], 
	\ns)=\bigO{\sqrt{\frac{\dims\log s}{\ns\numbits}\lor \frac{\log 
	\dims}{\ns}} },
	\end{align*}
	while the lower bounds given 
	in~\cref{eqn:gauss-est:ldp,eqn:gauss-est:ldp-largerho,eqn:gauss-est:comm} hold for 
	$\prm=\infty$, too.\footnote{That is, the upper and lower bounds only 
	differ by a $\log s$ factor for $\prm=\infty$ in the privacy case.}
\end{theorem}
\newest{We emphasize that, as discussed in~\cref{ssec:results,ssec:previous}, to the best of our knowledge~\cref{theorem:mean:estimation:gaussian} provides the first lower bounds for interactive Gaussian mean estimation under communication and privacy constraints.}
%
\begin{proof}[Proof of~\cref{theorem:mean:estimation:gaussian}]
	Let $\varphi$ denote the probability density function of the standard
	Gaussian distribution $\gaussian{\zero}{\II}$.  
	Fix $\prm\in[1,\infty)$. Let $\zdims=\dims$, $\cZ=\bool^\dims$, and 
	$\sparam=\frac{s}{2\dims}$; and suppose that, for some $\dst \in 
	(0,1/8]$, 
	there exists an $(\ns,\dst)$-estimator for $\mathcal{G}_{\dims,s}$ under 
	$\lp[\prm]$ loss.
	We fix a parameter $\gamma\eqdef \gamma(\prm) \eqdef 
	\frac{4\dst}{(s/2)^{1/\prm}}\in(0,1/2]$, and consider the set of
	distributions $\{\p_z\}_{z\in\cZ}$ of all $2^\dims$ spherical Gaussian
	distributions with mean  
	$\mu_z \eqdef \gamma (z+\one[\dims])$, where $z\in\cZ$. Again, note 
	that $\norm{\mu_z}_0 = \sum_{i=1}^\dims \indic{z_i=1} = 
	\numone{z}$, and~\cref{prop:numone} applies here too. Then by the 
	definition of
	Gaussian density, for $z\in\cZ$,  
	\begin{equation}
	\label{eq:gaussian-pdf}
	\p_z(x) = e^{-\gamma^2\normtwo{\mu_z}^2/2}\cdot 
	e^{\gamma\dotprod{x}{z+\one[\dims]}}\cdot \varphi(x).
	\end{equation}	
	Therefore, for $z\in\cZ$ and $i\in[\dims]$, we have
	\begin{equation}
	\label{eqn:gaussian_assn}
	\p_{z^{\oplus i}}(x) 
	=  e^{-2\gamma x_iz_i}e^{2\gamma^2z_i}\cdot\p_{z}(x)
	= \Paren{1+\phi_{z,i}(x)}\cdot\p_{z}(x),
	\end{equation}
	where 
	$\phi_{z,i}(x)\eqdef 1-e^{-2\gamma 
			x_iz_i}e^{2\gamma^2z_i}$. By using the 
	Gaussian moment-generating function, for $i\neq j$, 
	\[
	\bE{\p_z}{\phi_{z,i}(X)} = 0,\quad \bE{\p_z}{\phi_{z,i}(X)^2} = 
	\new{e^{4\gamma^2}-1}, \text{ 
		and } \bE{\p_z}{\phi_{z,i}(X)\phi_{z,j}(X)} = 0,
	\]
	so that~\cref{assn:decomposition-by-coordinates,assn:orthonormal}
	are satisfied \new{for $\kappavar\eqdef e^{4\gamma^2}-1$}. By our 
	choice of $\gamma$ and the assumption on $\dst$, 
	one can check that~\cref{assn:additive-loss} holds:  
	\[
	\lp[\prm](\mu(\p_z),\mu(\p_{z'})) 
	= 4\dst\mleft(\frac{\dist{z}{z'}}{\sparam\dims}\mright)^{1/\prm}.
	\]
	Moreover, similar to the product of Bernoulli case, using 
	\cref{prop:numone}, we can show that $\bP{Z}{\p_Z \in 
		\mathcal{G}_{\dims, s}} \le 1 - \tau/4$. This allows us to 
	apply~\cref{lem:mean:estimation}.
	
	\subsubsection{Privacy constraints \new{for $\priv \in (0,1)$}} For 
	$\cWpriv[\priv]$, upon 
	combining the
	bounds obtained by~\cref{cor:ldp,lem:mean:estimation}, we
	get
	\[
	\dims\leq
	112\ns\alpha^2(e^\priv -1)^2,
	\]
	whereby,
	upon noting that 
	$\alpha^2=e^{4\gamma^2}-1 \leq 8\gamma^2$ holds since $\gamma 
	\leq
	1/2$, and using the value of $\gamma=\gamma(\prm)$
	above, it follows that
	\[
	\dst^2\ge\frac{\dims (s/2)^{\frac 2 \prm}}{14336\cdot  \ns 
		(e^\priv 
		-1)^2}.
	\]
	Thus, $\cE_{\prm}(\mathcal{G}_{\dims,s}, \cWpriv[\priv], \ns)
	= \bigOmega{\sqrt{\frac{\dims s^{2/\prm}}{\ns\priv^2}} \land 1}$. This 
	establishes
	the lower bounds for $\cWpriv[\priv]$. (Recall that the bound for
	$\prm=\infty$ then follows from setting $\prm = \log \dims$.)
	
	\subsubsection{Communication constraints, \new{and privacy 
	constraints 
	for $\priv \ge 1$}} \new{For these cases, to prove a lower bound 
	with the desired dependence on $\priv$ or $\numbits$, we will need to use the 
	tighter bounds in \cref{cor:ldp,cor:simple-numbits} which 
	hold only under \cref{assn:subgaussianity}.}
	This, however, leads to an issue: the random vector
	$\phi_{z}(X) = (\phi_{z,i}(X))_{i\in[\dims]}$ is not subgaussian, due
	to the one-sided exponential growth, and
	therefore~\cref{assn:subgaussianity} does not hold.

	To overcome this and still obtain a linear dependence on $\numbits$ 
	\new{(or $\priv$)}
	(instead of the suboptimal $2^\numbits$ \new{(or $e^\priv$)}), we will 
	consider instead the class of ``truncated'' 
	Gaussian 
	distributions, whose 
	corresponding $\phi$ functions are subgaussian; and argue that these truncated distributions 
	are close enough to the original Gaussian distributions such a lower bound in the truncated case implies one in the original Gaussian case.
	
	In particular, we consider the following collection of truncated 
	Gaussian distributions. For $z \in \cZ$, let $\p_z$ be the density 
	function of a spherical Gaussian
	distribution with mean  
	$\mu_z$ as defined in \cref{eq:gaussian-pdf}. For a truncation bound 
	$\trunc$, let $\p_{z,B}$ be the distribution of $X \sim \p_z$ conditioned 
	on the event that $\norminf{X} \le B$. That is, we have, for $x\in\R^\dims$,
	\[
		\p_{z,B}(x) = C_z \p_z(x)
		\indic{\norminf{X} \le B},
	\]
	where $C_z = 1/\bP{X\sim\p_z}{\norminf{X} \le B}$. Then the following 
	bound follows from standard Gaussian concentration bound on each 
	dimension and a union bound over all dimensions.
	\begin{fact}
		Setting $B \eqdef 4\sqrt{\ln(\dims \ns)}$, we have, for every $z \in \cZ$,
		$
			\totalvardist{\p_{z,B}}{\p_z} \le \frac{1}{\dims^7 \ns^8}.
		$\cmargin{This is a bit overkill, I guess, but works.}
	\end{fact}
	Let $\p_{z,B}^{Y^\ns}$ be the distribution of the messages obtained by 
	executing the protocol when each user gets a sample from $\p_{z,B}$ 
	and let the corresponding mixtures be denoted by $\p_{+i, B}^{Y^\ns}$ 
	and $\p_{-i, B}^{Y^\ns}$. Then we have
	\begin{align*}
		 \totalvardist{\p_{+i}^{Y^\ns}}{\p_{-i}^{Y^\ns}} & \le  
		 \totalvardist{\p_{+i, B}^{Y^\ns}}{\p_{-i, B}^{Y^\ns}} + 
		\totalvardist{\p_{+i}^{Y^\ns}}{\p_{+i, B}^{Y^\ns}} + 
		\totalvardist{\p_{-i, B}^{Y^\ns}}{\p_{-i}^{Y^\ns}}  \\
		& \le \totalvardist{\p_{+i, B}^{Y^\ns}}{\p_{-i, B}^{Y^\ns}}  + \max_z 
		\left\{\totalvardist{\p_{z}^{Y^\ns}}{\p_{z, B}^{Y^\ns}} + 
			\totalvardist{\p_{z, B}^{Y^\ns}}{\p_{z}^{Y^\ns}}\right\} \\
			& \le \totalvardist{\p_{+i, B}^{Y^\ns}}{\p_{-i, B}^{Y^\ns}}  + 2 
			\max_z \totalvardist{\p_{z,B}^{\otimes n}}{\p_z^{\otimes n}} \\
			& \le \totalvardist{\p_{+i, B}^{Y^\ns}}{\p_{-i, B}^{Y^\ns}}  + 2 \ns 
			\max_z \totalvardist{\p_{z,B}}{\p_z} \\
			& \le  \totalvardist{\p_{+i, B}^{Y^\ns}}{\p_{-i, B}^{Y^\ns}}  + 
			\frac{2}{\dims^7\ns^7}.
	\end{align*}
	The third inequality follows from data processing inequality and the 
	fourth inequality follows from subadditivity of TV distance.
	
	Combining this with \cref{lem:mean:estimation}, for any protocol that 
	correctly learns the Gaussian family, we must have 
	\begin{equation}
		\label{eq:assouad-type:truncated}
		\frac{1}{\dims}\sum_{i = 1}^\dims \totalvardist{\p_{+i, 
		B}^{Y^\ns}}{\p_{-i, B}^{Y^\ns}} \ge 
		\frac{1}{8}.
	\end{equation}
	\new{Next we show that the $\phi$ functions corresponding to $\p_{z, B}$'s 
		are subgaussian and establish the corresponding upper bounds on the 
		average information 
		bound above. 		
		Note that
		\begin{align} \label{eq:phiB_def}
			\phi^B_{z,i}(x) & \eqdef \frac{\p^B_{z^{\oplus i}}(x)}{\p^B_{z}(x)}
			 - 1=  \frac{C_{z^{\oplus i}}}{C_z}e^{-2\gamma 
			x_iz_i}e^{2\gamma^2z_i} \indic{\norminf{x} \le B}- 1
		\end{align}
		By the inequality $|ab - 1| \le |a|\cdot |b-1| + |a-1|$, we have
		have, for all $z \in \cZ$,
		\begin{align*}
			\abs{\frac{C_{z^{\oplus i}}}{C_z} - 1}
			&\leq \frac{1}{C_z}\abs{C_{z^{\oplus i}}-1} + \abs{\frac{1}{C_z}-1} 
			\leq \abs{\frac1{\bP{X\sim\p_{z^{\oplus 
						i}}}{\norminf{X} \le B}} - 1} + \abs{\bP{X\sim\p_z}{\norminf{X} 
						\le 
				B} - 1}\\
			& \le  \frac{10}{d^7 n^7}.
		\end{align*}
		Moreover, for all $z \in \cZ$, for $\gamma \le \frac{1}{3B}$,
		\begin{align}
			\abs{e^{-2\gamma 
				x_iz_i}e^{2\gamma^2z_i} \indic{\norminf{x} \le B} - 1} \le 
				\abs{e^{2 \gamma^2 + 2 \gamma B} - 1} \le \abs{e^{3\gamma B} 
				- 1}  \le 6 \gamma B.
		\end{align}
		 Hence, applying the inequality $|ab - 1| \le |a|\cdot|b-1| + |a-1|$ again on 
		 \cref{eq:phiB_def}, we have for $\gamma \le \frac{1}{3B}$, 
		 \[
		 		|\phi^B_{z,i}(x) | \le 12 \gamma B +  \frac{10}{d^7 n^7}.
		 \]
		 Thus, we get that for all $z \in \cZ, i \in [\dims]$, $\phi^B_{z,i}$ is 
		 subgaussian with proxy $\sigma_B  = 12 \gamma B +  \frac{10}{d^7 
		 n^7}$. 
	 
	 	\noindent Under communication constraints, applying 
	 	\cref{cor:simple-numbits}, 
	 	we get 
	 	\[
	 		\Paren{\frac{1}{\dims}\sum_{i = 1}^\dims \totalvardist{\p_{+i, 
	 				B}^{Y^\ns}}{\p_{-i, B}^{Y^\ns}} }^2 \le \frac{14}{\dims}  
	 				\sigma_B^2 \ns 
	 				 \numbits.
	 	\]
	 	To conclude, we observe that by plugging our setting of $\gamma=\gamma(\prm)$ in the above inequality, we 
	 	must have
	 	\[
	 		\dst^2 \geq \frac{\dims (s/2)^{\frac 2 \prm}}{14336\cdot  
	 	\ns  \cdot B ^2 \numbits }
	 	\] in order to satisfy~\cref{eq:assouad-type:truncated}, 
 	hence proving the desired lower bound. The lower bound for LDP with 
 	$\priv > 1$ follows similarly by applying \cref{cor:ldp}.}
\end{proof}

%% file: sec-discrete.tex
\newcommand{\Ddims}{D}
We derive a lower bound for $\cE_\prm(\distribs{\dims}, \cW, \ns)$,
the minimax rate for discrete density estimation,
under local privacy and communication constraints.
\begin{theorem}
\label{theorem:mean:estimation:discrete}
Fix $\prm\in[1,\infty)$. For $\priv > 0$, and $\numbits \geq 1$, we 
have 
\begin{align}\label{eq:discrete:ldp}
\cE_\prm(\distribs{\dims}, \cWpriv[\priv], \ns) \gtrsim 
\sqrt{\frac{\dims^{2/\prm}}{\ns \Paren{(e^\priv - 1)^2 \land e^\priv}}\land 
\Paren{\frac{1}{\ns\Paren{(e^\priv - 1)^2 \land 
e^\priv}}}^{\frac{\prm-1}{\prm}}}\land 1
\end{align}
and
\begin{align}
\cE_\prm(\distribs{\dims}, \cWcomm[\numbits], \ns) \gtrsim  \sqrt{\frac{\dims^{2/\prm}}{\ns2^\numbits}\land \Paren{\frac{1}{\ns2^\numbits}}^{\frac{\prm-1}{\prm}}}\land 1\,.
\end{align}
\end{theorem}
In particular, for $\ns\Paren{(e^\priv - 1)^2 \land e^\priv} \geq \dims^2$ 
and $\ns(2^\numbits\land 
\dims) \geq \dims^2$, the first term of the corresponding lower bounds 
dominates. Before turning to the proof of this theorem, we provide two 
important corollaries; first, for the case of total variation distance 
($\lp[1]$), where combining it with known upper bounds we obtain 
optimal bounds for all $\priv > 0$. In particular, for $\priv\in(0,1]$ (high-privacy regime) we retrieve the lower bound established in~\cite{IIUIC}, which matches the upper bound from~\cite{AcharyaCFST21}. For $\priv > 1$ (low-privacy regime), our bound matches the upper bound for the 
noninteractive case, established in~\cite{YeB17,ASZ:18:HR}, showing that even in this low-privacy regime interactivity 
cannot lead to better rates, except maybe up to constant factors.
\begin{corollary}[Total variation distance]
\label{coro:mean:estimation:discrete:l1}
For $\priv \in (0,1]$, we have 
\begin{align}
\cE_1(\distribs{\dims}, \cWpriv[\priv], \ns) \asymp 
\sqrt{\frac{\dims^{2}}{\ns\priv^2}}\land 1
\end{align}
while, for $\priv > 1$, 
\begin{align}
\cE_1(\distribs{\dims}, \cWpriv[\priv], \ns) \asymp 
\sqrt{\frac{\dims^{2}}{\ns e^\priv}}\land 1\,.
\end{align}
For $\numbits \geq 1$, 
\begin{align}
\cE_1(\distribs{\dims}, \cWcomm[\numbits], \ns) \asymp \sqrt{\frac{\dims^{2}}{\ns (2^\numbits\land\dims)}}\land 1\,.
\end{align}
\end{corollary}
For the case of $\lp[2]$ estimation, we also obtain order-optimal bounds for local privacy, and for a wide range of $\numbits$ in the communication-constrained case:
\begin{corollary}[{$\lp[2]$ density estimation}] 
\label{coro:mean:estimation:discrete:l2}
For $\priv \in (0,1]$, we have 
\begin{align}
\cE_2(\distribs{\dims}, \cWpriv[\priv], \ns) \asymp 
\sqrt{\frac{\dims}{\ns\priv^2}}\land 
\sqrt[4]{\frac{\vphantom{\dims}1}{\ns\priv^2}}\land 1
\end{align}
while, for $\priv > 1$, 
\begin{align}
\cE_2(\distribs{\dims}, \cWpriv[\priv], \ns) \asymp 
\sqrt{\frac{\dims}{\ns e^\priv}}\land 
\sqrt[4]{\frac{\vphantom{\dims}1}{\ns e^\priv}}\land 1\,.
\end{align}
For $\numbits \geq 1$, 
\begin{align}
\cE_2(\distribs{\dims}, \cWcomm[\numbits], \ns)  \new{\asymp} 
\sqrt{\frac{\dims}{\ns (2^\numbits\land\dims)}}\land 
\sqrt[4]{\frac{\vphantom{\dims}1}{\ns (2^\numbits\land\dims)}}\land 1\,.
\end{align}
\end{corollary}
Note that for $\ns\Paren{(e^\priv - 1)^2 \land e^\priv}\geq \dims^2$ and 
$\ns(2^\numbits\land \dims) \geq \dims^2$, respectively, our bounds 
recover those obtained in~\cite{BCO:20} and~\cite{BHO:20} under these 
assumptions. 
We now establish~\cref{theorem:mean:estimation:discrete}.
\begin{proof}[Proof of~\cref{theorem:mean:estimation:discrete}]
Fix $\prm\in[1,\infty)$, and suppose that, for some $\dst \in (0,1/16]$, there exists an $(\ns,\dst)$-estimator for $\distribs{\dims}$ under $\lp[\prm]$ loss. Set
\[
    \Ddims \eqdef \dims \land \flr{\Paren{\frac{1}{16\dst}}^{\frac{\prm}{\prm-1}}}
\]
and assume, without loss of generality, that $\Ddims$ is even. By definition, we then have $\dst \in (0,1/(16\Ddims^{1-1/\prm})]$ and $D\leq \dims$; we can therefore restrict ourselves to the first $\Ddims$ elements of the domain, embedding $\distribs{\Ddims}$ into $\distribs{\dims}$, to prove our lower bound.

Let $\zdims=\frac{\Ddims}{2}$, $\cZ=\bool^{\Ddims/2}$, and $\sparam=\frac{1}{2}$; and suppose that, for some $\dst \in (0,1/(16\Ddims^{1-1/\prm})]$, there exists an $(\ns,\dst)$-estimator for $\distribs{\Ddims}$ under $\lp[\prm]$ loss. (We will use the fact that $\dst\leq 1/(16\Ddims^{1-1/\prm})$ for~\cref{eq:lp:construction:discrete} to be a valid distribution with positive mass, as we will need $|\gamma|\le \frac 1\Ddims$; and to bound $\kappavar$ later on, as we will require $|\gamma|\le \frac{1}{2\Ddims}$.)  
Define $\gamma=\gamma(\prm)$ as
\begin{equation}
        \label{eq:choice:gamma:p:discrete}
    \gamma(\prm) \eqdef \frac{4\cdot 2^{1/\prm} \dst}{\Ddims^{1/\prm}},
\end{equation}
which implies $\gamma\in[0,1/(2\Ddims)]$.
Consider the set of $\Ddims$-ary distributions $\cP^{\gamma}_{\rm Discrete}=\{\p_z\}_{z\in\cZ}$  defined as follows. 
For $z\in\cZ$, and $x\in\cX=[\Ddims]$
\begin{equation}
    \label{eq:lp:construction:discrete}
      \p_z(x) =
      \begin{cases}
            \frac{1}{\Ddims}+\gamma z_i, &\text{ if } x=2i,\\
            \frac{1}{\Ddims}-\gamma z_i, &\text{ if } x=2i-1.\\
      \end{cases}
\end{equation} 
 For $z\in\cZ$ and $i\in[\Ddims/2]$, we have 
 \begin{align}
    \p_{z^{\oplus i}}(x) 
    &=  \Paren{1-\frac{2\Ddims\gamma z_i}{1+\Ddims\gamma z_i}\indic{x=2i}+\frac{2\Ddims\gamma z_i}{1-\Ddims\gamma z_i}\indic{x=2i-1} }\p_{z}(x)
\nonumber
    \\
    &= \Paren{1+\phi_{z,i}(x)}\p_{z}(x),
    \label{eqn:discrete_assn}
 \end{align}
 where
 \[
\phi_{z,i}(x)\eqdef z_i\cdot\frac{2\Ddims\gamma}{1-\Ddims^2\gamma^2}\Paren{ (1+\Ddims\gamma z_i)\indic{x=2i-1}-(1-\Ddims\gamma z_i)\indic{x=2i}}.
 \]
Once again, we can verify that for $i\neq j$
 \[
    \bE{\p_z}{\phi_{z,i}(X)} = 0,\quad \bE{\p_z}{\phi_{z,i}(X)^2} = \new{\frac{8\gamma^2\Ddims}{1-\gamma^2\Ddims^2}}, \text{ and } \bE{\p_z}{\phi_{z,i}(X)\phi_{z,j}(X)} = 0, 
 \]
 so that~\cref{assn:decomposition-by-coordinates,assn:orthonormal} are satisfied \new{for $\kappavar\eqdef 16\gamma^2\Ddims$ (using that $\Ddims\gamma\leq 1/2$ to simplify the bound)}.\footnote{It is worth noting that~\cref{assn:subgaussianity} will not hold for any useful choice of the subgaussianity parameter.} Thus, we can invoke the first part of~\cref{thm:avg:coordinate}. Note that~\cref{assn:additive-loss} holds, since  
$\lp[\prm](\p_z,\p_{z'}) = \gamma\dist{z}{z'}^{1/\prm} = 4\dst\mleft(\frac{\dist{z}{z'}}{\tau\Ddims}\mright)^{1/\prm}$.
Therefore, we can apply~\cref{lem:mean:estimation} as well.

For $\cWpriv[\priv]$, by combining the bounds obtained 
by~\cref{cor:ldp,lem:mean:estimation}, we get
\[
\Ddims\leq \new{56}\ns\kappavar\Paren{(e^\priv -1)^2 \land e^\priv},
\]
whereby, upon recalling the value of $\kappavar$ 
and using the setting of $\gamma=\gamma(\prm)$ from~\cref{eq:choice:gamma:p:discrete},
it follows that
\[
\dst^2\geq\frac{\Ddims^{\frac{2}{\prm}}}{7168\cdot 2^{2/\prm}\cdot 
\ns\Paren{(e^\priv -1)^2 \land e^\priv}} \asymp \frac{\dims^{2/\prm}\land 
\dst^{-2/(\prm-1)}}{\ns\Paren{(e^\priv -1)^2 \land e^\priv}}.
\]
Thus we obtain the bound \cref{eq:discrete:ldp} as claimed.

Similarly, for $\cWcomm[\numbits]$, upon combining the bounds obtained by~\cref{cor:simple-numbits,lem:mean:estimation} and recalling that $\abs{\cY}=2^\numbits$, we get
\[
\dst^2\geq\frac{\Ddims^{\frac{2}{\prm}}}{7168\cdot2^{2/\prm}\cdot \ns 2^\numbits},
\]
which gives
$\cE_{\prm}(\distribs{\Ddims}, \cWcomm[\numbits], \ns) = \bigOmega{\sqrt{\frac{\dims^{2/\prm}}{\ns 2^\numbits} \land \Paren{\frac{1}{\ns 2^\numbits}}^{\frac{\prm-1}{\prm}} } }$,\footnote{Finally, note that we could replace the quantity $2^\numbits$ above by $2^\numbits\land \dims$, or even $2^\numbits\land \Ddims$, as for $2^\numbits \geq \Ddims$ there is no additional information any player can send beyond the first $\log_2\Ddims$ bits, which encode their full observation. However, this small improvement would lead to more cumbersome expressions, and not make any difference for the main case of interest, $\prm=1$.} concluding the proof.
\end{proof}

%% file: sec-fullinteractive.tex
In this appendix, we describe how to extend our results, presented in the sequentially interactive model, to the more general interactive setting. We first formally define this setting and the corresponding notion of protocols. Hereafter, we use ${}^\ast$ for the Kleene star operation, \ie $V^\ast = \bigcup_{n=0}^\infty V^n$.
\begin{definition}[Interactive Protocols]
Let $X_1, \dots, X_{\ns}$ be \iid samples from $\p_\theta$, $\theta\in 
\Theta$, and $\cW^\ast$ be a collection of sequences of pairs of channel 
families and players; that is, each element of $\cW^\ast$ is a sequence
$
  (\cW_{t}, j_t)_{t\in\N}
$
where $j_t\in[\ns]$. 
An \emph{interactive protocol $\Pi$ using $\cW^\ast$} comprises a
random variable $U$ (independent of the input $X_1, \dots, X_{\ns}$) and, for each $t\in\N$,
mappings 
\begin{align*}
  \sigma_t&\colon Y_1, \ldots, Y_{t-1}, U\mapsto N_{t} \in [\ns]\cup \{\bot\}\\
  g_t&\colon Y_1, \ldots, Y_{t-1}, U\mapsto W_t
\end{align*}
with the constraint that $((W_1,N_1),\dots,(W_t,N_t))$ must be consistent 
with some sequence from $\cW^\ast$; that is, there exists 
$((\cW_{s},j_s))_{s\in\N}\in \cW^\ast$ such that $W_s \in \cW_{s}$ and 
$N_s=j_s$ for all $1\leq s\leq t$. 
These two mappings respectively indicate (i)~whether the protocol is to stop (symbol $\bot$), and, if not, which player is to speak at round $t\in\N$, and (ii))~which channel this player selects at this round.  

In round $t$, if $N_t=\bot$, the protocol ends. Otherwise, player $N_t$ (as selected by the protocol, based on the previous messages) 
uses the channel $W_t$ to produce the message (output) $Y_t$
according to the probability measure $W_t(\cdot\mid X_{N_t})$. We further require that $T \eqdef \inf\setOfSuchThat{t\in\N}{ N_t = \bot }$ is finite a.s.  The messages
$Y^T=(Y_1, \ldots, Y_{T})$ received by the referee and the
public randomness $U$ constitute the \emph{transcript} of the protocol $\Pi$.
\end{definition}

In other terms, the channel used by the player $N_t$ speaking at time $t$ is a Markov kernel
\[
    W_t\colon \saY_t\times \cX\times \cY^{t-1}\to[0,1]\,,
\]
with $\cY_t \subseteq \cY$; 
and, for player $j\in[\ns]$, the allowed subsequences $(\cW_{t},j_t)_{t\in\N: j_t=j}$ capture the possible sequences of channels allowed to the player. As an example, if we were to require that any single player can speak at most once, then for every $j\in[\ns]$ and every $(\cW_{t}, j_t)_{t\in\N}\in\cW^\ns$, we would have
$
    \sum_{t=1}^\infty \indic{j_t=j}\leq 1
$.

In the interactive model, we can then capture the constraint that each player must communicate at most $\numbits$ bits in total by letting $\cW^\ns$ be the set of sequences $(\cW_t^{{\rm comm},\numbits_t}, j_t)_{t\in\N}$ such that
\[
    \forall j\in[\ns], \qquad  \sum_{t=1}^\infty \numbits_t \cdot \indic{j_t=j} \leq \numbits\,.
\]
In the simpler sequentially interactive model, this condition simply 
becomes the choice of $\cW^\ns=(\cW^{{\rm 
comm},\numbits},\dots,\cW^{{\rm comm},\numbits})$.

\subsection{Lower Bounds under Full Interactive Model}

Next we discuss how our technique extends to the full interactive model. 
For any full interactive protocol $\Pi$, let $Y^\ast \in \cY^\ast$ be the 
message 
sequence generated by the protocol. Then, for all $y^\ast \in 
\cY^\ast$, we have
\[
	\probaDistrOf{X^n \sim \p}{Y^\ast = y^\ast} = \bE{X^\ns \sim  \p}{ 
	\prod_{t = 
		1}^\infty 
	W_t\Paren{y_t \mid 
		X_{\sigma_t(y^{t-1})}, {y^{t-1}}}}.
\]

The following lemma states that if $X^\ns$ are 
generated from a 
product distribution, the distribution of the 
transcript satisfies a property similar to the 
``cut-and-paste'' property from~\cite{BarYossefJKS04}.
\begin{lemma}[\cite{HOW:18:v1}] \label{lem:cut-paste}
	If $X^\ns \sim \p = \otimes_{t = 1}^\ns \p_t$, the transcript of the 
	protocol 
	satisfies
	\begin{align} \label{eq:cut-paste}
		\probaDistrOf{X^\ns \sim \p}{Y^\ast = y^\ast} = \prod_{t=1}^{\ns} 
		\bE{X_t
		\sim 
		\p_t}{\ff_t(y^\ast, X_t)},
	\end{align}
	where
	$
	\ff_t(y^\ast, x_t) = \prod_{j = 
		1}^\infty 
	W_j(y_j \mid 
	x_t, {y^{j-1}}) \indic{\sigma_j(y^{j-1}) = t}.
	$
\end{lemma}
Hence, when $X^\ns \sim \p_z^{\otimes \ns}$ we have
\[
	\p_z^{y^\ast} \eqdef \probaDistrOf{X^n \sim \p_z^{\otimes 
	n}}{Y^\ast = y^\ast} = \prod_{t = 1}^{\ns} \bE{X_t \sim 
\p_z}{\ff_t(y^\ast, X_t)}.
\]
Here we can define a similar notion of ``channel'' for a communication 
protocol 
$\Pi$ for the $i$th player when the underlying distribution is $\p_z$ by setting 
\begin{equation} \label{eqn:channel}
\channel_{t, \p_z}(y^\ast \mid x) = \ff_t(y^\ast, x) \Paren{\prod_{j \neq t} 
	\bE{X_{j} 
		\sim 
		\p_z}{\ff_{j}(y^\ast, X_{j})}}.
\end{equation}
Then we have, for all $t \in [\ns]$,
\[
\bE{X_t 
	\sim 
	\p_z}{\channel_{t, \p_z}(y^\ast\mid X_t) } = \bP{X^\ns \sim \p_z^{\otimes 
		n}}{Y^\ast = y^\ast}.
\]
We proceed to prove a bound similar to 
\cref{lemma:per:coordinate} in terms of the ``channel'' defined in 
\cref{eqn:channel}, as stated below. 

\begin{theorem}[Information contraction bound]
	\label{lemma:per:coordinate:full:interactive} 
	Fix $\sparam\in(0,1/2]$. Let $\Pi$ be a fully interactive
	protocol using $\cW^\ns$, and let $Z$ be a random variable on $\cZ$ 
	with distribution $\rademacher{\sparam}^{\otimes\zdims}$. Let 
	$(Y^\ast,U)$ be the transcript of $\Pi$ when the input
	$X_1, \ldots, X_\ns$ is i.i.d.\ with common distribution $\p_Z$.
	Then, 
	under~\cref{assn:decomposition-by-coordinates},
	\begin{align} 
	&\lefteqn{\Paren{\frac{1}{\zdims}\sum_{i=1}^\zdims\totalvardist{\p_{+i}
				^{Y^\ast}}{\p_{-i}^{Y^\ast}}}^2} \nonumber 
	\\
	\nonumber &\le 
	\frac{7}{\zdims} \alpha^2  
	\sum_{j=1}^\ns\max_{z\in\cZ}\max_{ (\cW_{t}, j_t)_{t\in\N} \in 
	\cW^\ns} \sum_{i=1}^\zdims 
	\int_{y^\ast \in \cY^\ast} \frac{\bE{\p_{z}}{\phi_{z,i}(X)\channel_{j, 
	\p_z}(y^\ast\mid 
	X)}^2}{\bE{\p_{z}}{\channel_{j, \p_z}(y^\ast\mid X)}} \dd{\mu} \,, 
	\end{align}
	where $\p_{+i}^{Y^\ast}\eqdef \bEEC{\p_Z^{Y^\ast}}{Z_i=1}$, 
	$\p_{-i}^{Y^\ast}\eqdef \bEEC{\p_Z^{Y^\ast}}{Z_i=1}$.
\end{theorem}

We can see the bound is in identical form to \cref{lemma:per:coordinate} 
except that we replace each player's channel with the $\channel_{j, 
\p_z}(y^\ast\mid 
X)$ we defined. Other similar bounds in \cref{sec:general-bound} 
can also be derived under additional assumptions and specific constraints. 
We present 
the proof for \cref{lemma:per:coordinate:full:interactive} below and omit 
the detailed statements and proof for other bounds.
\begin{proof}
	Analogously to \cref{eq:decomposition}, we can get\todonote{Doublecheck this.}
	\begin{equation} \label{eq:tv2hell:full}
  \frac{1}{\zdims}\Paren{\sum_{i=1}^\zdims \totalvardist{\p_{+i}^{Y^\ast}}{\p_{-i}^{Y^\ast}}}^2
  \leq 14 \sum_{t=1}^\ns \bE{Z}{ \sum_{i=1}^\zdims \hellinger{\p^{Y^\ast}_{Z}}{\p^{Y^\ast}_{t \gets Z^{\oplus i}}}^2 } 
  \end{equation}
	
	For all $z \in \bool^\zdims$ and $i,t$, by the definition of Hellinger distance and~\cref{eq:cut-paste}, we have
	\begin{align*}
	2\hellinger{\p^{Y^\ast}_{z}}{\p^{Y^\ast}_{t \gets z^{\oplus i}}}^2
	&= \int_{ y^\ast \in \cY^\ast} \prod_{\substack{1\leq j\leq \ns\\j\neq 
	t}}\bE{X_{j} \sim \p_z}{\ff_{j}(y^\ast, X_{j})} \Paren{\sqrt{\bE{X_{t} \sim 
	\p_{z^{\oplus i}}}{\ff_{t}(y^\ast, X_{t})}} - \sqrt{\bE{X_{t} \sim 
	\p_z}{\ff_{t}(y^\ast, X_{t})}}}^2 \dd{\mu} \nonumber  \\
	& \leq\int_{ y^\ast \in \cY^\ast} \Big(\prod_{j\neq t}\bE{X_{j} \sim 
			\p_z}{\ff_{j}(y^\ast, X_{j})} \Big) \Paren{\frac{(\bE{X_t \sim 
				\p_z}{\ff_t(y^\ast, X_t)} - \bE{X_t \sim 
				\p_{z^{\oplus i}}}{\ff_t(y^\ast, X_t)} )^2}{\bE{X_t \sim 
				\p_{z}}{\ff_t(y^\ast, X_t)} }} \dd{\mu} ,\label{eqn:square}
	\end{align*}
	Proceeding from above, we get 
	under~\cref{assn:decomposition-by-coordinates},
\begin{align*}
2\hellinger{\p^{Y^\ast}_{z}}{\p^{Y^\ast}_{t \gets z^{\oplus i}}}^2
 & \le \alpha^2\int_{ y^\ast \in \cY^\ast} \Paren{\prod_{j \neq t}\bE{X_{j} 
 \sim \p_z}{\ff_{j} (y^\ast, X_{j} )}} 
\Paren{\frac{\bE{X_t \sim \p_{z}}{\phi_{z,i}(X_t)\ff_t(y^\ast, X_t)}^2}{\bE{X_t 
\sim 
			\p_{z}}{\ff_t(y^\ast, X_t)} }}  \dd{\mu}\\
		&= \alpha^2\int_{ y^\ast \in \cY^\ast} 
\frac{\bE{X_t \sim \p_{z}}{\phi_{z,i}(X_t)\ff_t(y^\ast, X_t)\prod_{j \neq 
t}\bE{X_{j} \sim \p_z}{\ff_{j} (y^\ast, X_{j} )}}^2}{\bE{X_t \sim 
\p_{z}}{\ff_t(y^\ast, X_t) \prod_{j \neq t}\bE{X_{j} \sim \p_z}{\ff_{j} (y^\ast, 
X_{j} )}} }  \dd{\mu}\\
		&= \alpha^2  \int_{ y^\ast \in 
						\cY^\ast} 
						\frac{\bE{X_t \sim \p_{z}}{\phi_{z,i}(X_t)\channel_{t, 
						\p_z}(y^\ast\mid 
									X)}^2}{\bE{X_t \sim 
									\p_{z}}{\channel_{t, \p_z}(y^\ast\mid 
									X) }} \dd{\mu}.
\end{align*}
Plugging the above bound into \cref{eq:tv2hell:full}, we can obtain the bound 
in~\cref{lemma:per:coordinate:full:interactive} by taking the maximum 
over all $z \in \bool^\zdims$ and all possible channel sequences.
\end{proof}

%% file: sec-measure-change.tex
We here provide a variant of Talagrand’s transportation-cost inequality
which is used in deriving~\cref{eqn:subgaussian-bound}
(under~\cref{assn:subgaussianity})
in the second part of~\cref{thm:avg:coordinate}.
We note that this type of result is not novel, and can be derived from standard arguments in the literature (see, e.g.,~\cite[Chapter~8]{Boucheron:13} or~\cite[Chapter~4]{vanHandel:16}). However, the lemma below is specifically tailored for our purposes, and we provide the proof for completeness. A similar bound was derived in~\cite{ACT:20}, where Gaussian mean testing under communication constraints was considered. 
\new{
\begin{lemma}[A measure change bound]\label{l:basic_mc}
Consider a random variable $X$ taking values in $\cX$
and with distribution $P$. Let
$\Phi\colon\cX\to\R^\zdims$ be such that the random vector
$\Phi(X)$ is $\sigma^2$-subgaussian.
 Then, for any
function $a\colon\cX\to[0,\infty)$ such that $\bEE{a(X)}<\infty$, we have
\[
\frac{\normtwo{\bE{}{\Phi(X)a(X)}}^2}{\bE{}{a(X)}^2}\leq 2\sigma^2
\frac{\bE{}{a(X)\ln a(X)}}{\bE{}{a(X)}} + 2\sigma^2\ln \frac{1}{\bE{}{a(X)}}.
\]
\hnote{Commenting the part below since it is not serving  any purpose now.}
\end{lemma}
\begin{proof}
By an application of Gibb's variational principle ($cf.$~\cite[Corollary~4.14]{Boucheron:13}) the following holds:
For a random variable
$Z$ and distributions $P$ and $Q$ on the underlying
probability space satisfying $Q\ll P$ (that is, such that $Q$ is absolutely continuous with respect to $P$),
we have
\[
\lambda \bE{Q}{Z} \leq \ln \bE{P}{e^{\lambda Z}}+ \kldiv{Q}{P}.
\] 
To apply this bound, set $P$ to be the distribution of $X$
and let $Q\ll P $ be defined
using its density (Radon--Nikodym derivative) with respect to $P$
given by 
\[
\dv{Q}{P} =  \frac{a(X)}{\bE{P}{a(X)}}.
\]
Now, note that for any unit vector $v$, we have, setting $Z= v^\transp \Phi(X)$ and
 using the $\sigma^2$-subgaussianity of $\Phi(X)$, that
\[
\lambda \bE{Q}{v^\transp \Phi(X)} \leq \ln \bE{P}{e^{\lambda v^\transp \Phi(X) }}+\kldiv{Q}{P} \leq
\frac{\sigma^2\lambda^2}{2}+ \kldiv{Q}{P}.
\]
In particular, for $\lambda = \frac{1}{\sigma}\sqrt{2\kldiv{Q}{P}}$, we get
\[
\bE{Q}{v^\transp \Phi(X)} \leq \sigma\sqrt{2\kldiv{Q}{P}}.
\]
Applying this to the unit vector $v \eqdef \frac{\bE{Q}{\Phi(X)}}{\normtwo{\bE{Q}{\Phi(X)}}}$ then yields
\[
\normtwo{\bE{Q}{\Phi(X)}} \leq \sigma\sqrt{2\kldiv{Q}{P}}.
\]
To conclude, it then suffices to observe that
\begin{align*}
\kldiv{Q}{P} &= \frac{\bE{P}{a(X)\ln a(X)}}{\bE{P}{a(X)}} + \ln \frac{1}{\bE{P}{a(X)}}.
\end{align*}
The proof is completed by combining the bounds above, as 
$
\bE{Q}{\Phi(X)} = \frac{\bE{P}{\Phi(X)a(X))}}{\bE{P}{a(X)}}.$
\end{proof}
}

%% file: sec-bernoulli-ub-sparse-interactive.tex
Recall that $\mathcal{B}_{\dims,\sprs}$, the family of 
$\dims$-dimensional $\sprs$-sparse product Bernoulli distributions, is defined as
\begin{equation}\label{def:sparse:product:bernoulli}
    \mathcal{B}_{\dims,\sprs} \eqdef \setOfSuchThat{ 
    \bigotimes_{j=1}^\dims\rademacher{\frac{1}{2}(\mu_j+1)} }{ 
    \mu\in[-1,1]^\dims, \norm{\mu}_0 \leq \sprs}\,.
\end{equation}

We now provide the interactive protocols achieving the upper bounds of~\cref{theorem:mean:estimation:bernoulli} for sparse product Bernoulli mean
estimation under LDP and communication constraints . 

Our protocols has two ingredients described below: 

\begin{itemize}
	\item [{\bf 1. Estimating non-zero mean coordinates.}] In this step we will start with $S_0=[d]$, the set of all possible coordinates. Then we will iteratively prune the set $S_0\to S_1\to\ldots\to S_T$, such that $|S_T|=3s$ (this step is skipped if $\sprs\ge\dims/3$) is a good estimate for the set of coordinates with non-zero mean. 
	\item [{\bf 2. Estimating the non-zero means.}] We then estimate the means of the coordinates in $S_T$, which is equivalent to solving a dense mean estimation problem in $3s$ dimensions.
\end{itemize}

In the next two sections, we provide the details of the algorithm that matches the lower bounds obtained
in~\cref{sec:bernoulli-product} for interactive protocols under LDP  and communication constraints respectively.

\input{sec-bernoulli-ub-sparse-ldp}
\input{sec-bernoulli-ub-sparse-comm}

%% file: sec-bernoulli-ub-sparse-ldp.tex
\subsubsection{LDP constraints}
\newchange{In this subsection, we will focus on the case $\priv \in (0, 1]$ (high-privacy regime). For 
the case $\priv > 1$, we rely a privatization of the 
communication-limited 
algorithm, which will be discussed at the end of 
\cref{sec:communication-upper}.} 
Our protocol for Bernoulli mean estimation under LDP constraints is 
described in~\cref{alg:ldp}. As stated above, in each round $t=1, \ldots, T$, 
for each $j\in S_{t-1}$ a new group of players apply the well known binary 
Randomized Response (RR) 
mechanism~\cite{Warner:65,KLNRS:11} to their $j$th coordinate. Using 
these messages we then guess a set of coordinates with highest possible 
means (in absolute value) and prune the set to $S_t$. This is done in 
Lines~\ref{state-3}-\ref{state-6} of~\cref{alg:ldp}.

In Lines \ref{state-7}-\ref{state-11}, the algorithm uses the same approach 
to estimate the means of coordinates within $S_T$ and sets remaining 
coordinates to zero.

The privacy guarantee
follows immediately from that of the RR
mechanism, and further, this only requires one bit of communication
per player.  

\begin{algorithm}[h]
	\caption{LDP protocol for mean estimation for the product of Bernoulli family}
	\label{alg:ldp}
	\begin{algorithmic}[1]
		\Require $\ns$ players, dimension $\dims$, sparsity parameter 
		$\sprs$, privacy parameter $\priv$.
		\State Set $\numT \eqdef \log_3 \frac{\dims}{3\sprs}$, $\alpha \eqdef 
		\frac{e^\priv}{1 + e^\priv}$, $\setS_0 = [\dims]$, $N_0 \eqdef
		\frac{\ns}{6\dims}$. 
		\For {$t = 1, 2, \ldots, \numT$}\label{state-3}
			\For{$j\in S_{t-1}$}
			\State Get a group of new players 
			$G_{t,j}$ of size $N_t = N_0 \cdot 2^t$. 
			\State Player 
			$i\in G_{t,j}$, upon observing $X_i\in\bool^\dims$ sends the message $Y_i\in\bool$ such that
			\begin{equation} \label{eqn:rr}
			Y_i =
			\begin{cases}
			(X_i)_j & \text{w.p. $\alpha$,}	\\
			-(X_i)_j & \text{w.p. $1-\alpha$.}	
			\end{cases}
			\end{equation}\label{state-4}
		\State Set $M_{t,j} \eqdef \sum_{i\in G_{t,j}} Y_i$. Let $S_{t} \subseteq S_{t-1}$ be the 
		set of the $|S_{t-1}|/3$ indices with the largest $|M_{t,j}|$.\label{state-6}
		\EndFor
		\EndFor
		\For {$j\in S_T$}\label{state-7}
		\State Get a group of new players $G_{T,j}, j 
		\in S_{\numT}$ of size $N_{\numT+1}= N_0 \cdot 
		2^\numT$. 
		\State Player 
		$i\in 
		G_{T,j}$, sends 
		the message $Y_i\in\bool$ according to~\cref{eqn:rr} and $M_{T,j} \eqdef \sum_{i\in G_{T,j}} Y_i$\label{state-8}
		\EndFor
		
\For{$j\in[d]$}
\State\begin{equation*}
			\widehat{\mu}_j =
			\begin{cases}
			\frac{M_{j,\numT}}{(2\alpha - 1) N_{\numT+1}} & \text{if $j\in S_T$,}	\\
			0 & \text{otherwise.}	
			\end{cases}
			\end{equation*}
\EndFor
		\State \Return $\widehat{\mu}$.\label{state-11}
		\end{algorithmic}
\end{algorithm}

The performance guarantee of \cref{alg:ldp} is stated below, which matches the 
lower 
bounds obtained
in~\cref{sec:bernoulli-product}.
\begin{proposition}
	\label{theorem:mean:estimation:bernoulli:ub:ldp}
	Fix $\prm\in[1,\infty]$. For $\ns \geq 1$ and $\priv \in (0,1]$, \
	\cref{alg:ldp} is an $(\ns, \dst)$-estimator using $\cW_\priv$ under 
	$\ell_\prm$ loss for 
	$\mathcal{B}_{\dims,\sprs}$ with 
	$
		\dst = \bigO{\sqrt{\frac{\prm \dims \sprs^{2/\prm}}{\ns \priv^2}}}
	$ for $\prm \le 2 \log \sprs$ and $
	\dst = \bigO{\sqrt{\frac{\dims \log \sprs  }{\ns \priv^2}}}
	$ for $\prm >  2 \log \sprs$.
\end{proposition}

\begin{proof}
The total number of players used by \cref{alg:ldp} uses is
	\[
		\sum_{t = 1}^{\numT+1} |\setS_{t-1}|\cdot N_t = |\setS_0|\cdot N_0 \cdot\sum_{t = 
		1}^{\numT+1} \frac{2^t}{3^{t-1}} \le 6 |\setS_0| \cdot N_0  = \ns.
	\]
To prove the utility guarantee, we bound the estimation error in the estimated set $\setS_{\numT}$ and the 
	error outside the set $\setS_{\numT}$ in the following lemma. 

\begin{lemma}\label{lem:err:withinS}
	Let $\setS_\numT$ be the subset obtained from the first stage of 
	\cref{alg:ldp}. Then, 
	\[
	\max\left\{\bEE{\sum_{j \notin \setS_\numT} \abs{\mu_j - \widehat{\mu}_j}^\prm},\bEE{\sum_{j \in \setS_\numT} \abs{\mu_j - \widehat{\mu}_j}^\prm}\right\} = 
	\bigO{\sprs \Paren{\frac{\prm\dims}{\ns 
				\priv^2}}^{p/2}}.
	\]
\end{lemma}

The proposition follows directly from the lemma.
Indeed, for $\prm 
> 2 \log \sprs$,  by monotonicity of $\lp[\prm]$ norms we have $\norm{\mu-\hat{\mu}}_{\prm} \leq \norm{\mu-\hat{\mu}}_{\prm'}$ for 
all 
$\prm' \leq \prm$, and thus choosing $\prm' \eqdef 2\log \sprs$ is sufficient to 
obtain the stated bound.
\end{proof}

\begin{proof}[Proof of \cref{lem:err:withinS}]
We prove the bound on each term individually. The first term captures the performance of our estimator within coordinates in $S_T$ and the second term states that we do not ``prune'' too many coordinates with high non-zero means. 

\medskip
\noindent\textbf{Bounding the first term.}
For $j \notin \setS_\numT$, we output $\widehat{\mu}_j = 0$. Therefore,
	\[
		\bEE{\sum_{j \notin \setS_\numT} \abs{\mu_j -\widehat{\mu}_j}^\prm} = 
		\sum_{j}  \bEE{\abs{\mu_j - 
				\widehat{\mu}_j}^\prm \cdot \indic{j \notin \setS_\numT}} = \sum_{j}  \abs{\mu_j}^\prm \cdot \bPr{j \notin \setS_\numT}.
	\]
	Since $\mu$ is $\sprs$-sparse, it will suffice to show that for all $j$ with $|\mu_j| > 0$,
	\begin{equation} \label{eqn:prob:missing}
	|\mu_j|^\prm \cdot \bPr{j \notin \setS_\numT}= \bigO{ 
		\Paren{\frac{\prm \dims}{\ns 
				\priv^2}}^{\prm/2}}.
	\end{equation}
	
Let \[
\Th := 20\sqrt{\frac{\dims }{\ns 
			(2\alpha-1)^2}}.
\]
Note that for $\priv\in(0,1]$, we have	$2\alpha - 1 \geq \frac{e-1}{e+1}\priv$. Therefore, if $|\mu_j| \le 
	\Th$, then~\cref{eqn:prob:missing} holds since 
			$\bPr{j 
			\notin \setS} \le 1$. We 
			hereafter assume $|\mu_j| > 
	\Th$, and let $\mu_j = \beta_j \Th$ 
	with $\beta_j > 1$. Let $E_{t,j}$ be the event that coordinate $j$ is 
	removed in round $t$ given that $j \in \setS_{t-1}$. Then we have
	\[
		\bPr{j \notin \setS_\numT} \le \sum_{t = 1}^\numT \bPr{E_{t,j}}.
	\]

We proceed to bound each $\bPr{E_{t,j}}$ separately. 
Note that for $i \in G_{t,j}$, $Y_i\in\bool$ and by~\cref{eqn:rr}
\begin{equation}
\label{eqn:exp-Y}	
\expect{Y_i} = 
(2 \alpha - 1)\cdot \mu_j = (2 \alpha - 1)\beta_j 
	\Th.
\end{equation}
 	
	Let $\numlc_{t,j}$ be the number 
	of coordinates $j'$ with $\mu_{j'} = 0$ and 
	$|M_{t,j'}| \ge \frac{1}{2}N_t (2 \alpha - 1)\beta_j\Th$. 
	Since we select the $|S_{t-1}|/3$ coordinates with the largest 
	magnitude of the sum, for 
	$j \notin \setS_{t}$ to happen at least one of the 
	following must occur: (i)~$\numlc_{t,j} > \frac{1}{3}|S_{t-1}| - \sprs$, or (ii)~$M_{t,j} < \frac{1}{2}N_t  (2\alpha-1)\beta_j\Th$. 

	By Hoeffding's inequality, we have
	\[
		\bPr{M_{t,j} < \frac{1}{2}N_t(2\alpha -1 )\beta_j\Th} 
		\le 
		\exp\Paren{ - \frac{1}{8} N_t\Paren{(2\alpha-1) 
		\beta_j\Th}^2} < \exp\Paren{ - 5\cdot 2^t \beta_j^2}
		.
	\]
 	Let $p_{t,j} \eqdef e^{- 5\cdot 2^t \beta_j^2}$. Similarly, for 
 	any 
 	$j'$ such that $\mu_{j'} = 0$,  
	\[
		\bPr{|M_{t,j'}| \ge \frac{1}{2}N_t (2\alpha -1 ) \beta_j\Th} \le  
		2p_{t,j}.
	\]
Since all coordinates are independent, $\numlc_{t,j}$ is 
	binomially distributed with mean at most $2 p_{t,j} |S_{t-1}|$. By Markov's inequality, 
	we get
	\[
		\bPr{\numlc_{t,j}> \frac{1}{3}|S_{t-1}| - \sprs} \le 
		\frac{\bEE{\numlc_{t,j}}}{|S_{t-1}|/3 - \sprs} \le 
		p_{t,j},
	\]
	recalling that $|S_{t-1}| = \dims 3^{t-1} \geq 9\sprs$.
	By a union bound and summing over $t \in [\numT]$, we get
	\[
	\bPr{j \notin \setS_\numT} \le \sum_{t = 1}^\numT \bPr{E_{t,j}} 
	\le\sum_{t = 1}^\numT 3 p_{t,j} = 3\sum_{t = 1}^\numT 
	\exp\Paren{ - 2^t \cdot 5 \beta_j^2} \le 6 \exp\Paren{ - 5 
	\beta_j^2} .
	\]
	Not that for $x > 0$, $x^\prm e ^{-x^2} \le 
	\Paren{\frac{\prm}{2e}}^{\prm/2}$. 
	Hence
	\[
		|\mu_j|^\prm \cdot \bPr{j \notin \setS_\numT} \le 6 H^\prm 
		\beta_j^{\prm}e^{-5\beta_j^2} \leq 
	\Paren{C\frac{\prm \dims }{\ns \priv^2}}^{\prm/2},
	\]
	for some absolute constant $C>0$, completing the proof.

\medskip
\noindent\textbf{Bounding the second term.}
Note that $S_\numT$ is a random variable 
	itself. We show that the 
	bound holds for any realization of $S_\numT$.
	 We need the following result which follows from standard moment bounds on binomial 
	distributions. 

\begin{fact}
		\label{fact:binomial-lp}
		Let $\prm \geq 1$, $m\in\N$, $0\le q\le 1$, and $N\sim\binomial{m}{q}$. Then,
		$
		\expect{\abs{N- mq}^{\prm}}\le 2^{-\prm/2}m^{\prm/2} 
		\prm^{\prm/2}
		$ .
	\end{fact}

Applying this with $m= N_T \ge \frac{\ns}{6\dims}$, the transformation 
	from Bernoulli 
	to $\{-1,+1\}$, and the scaling by $2\alpha-1$, yields for $j \in 
	\setS_\numT$, and using~\cref{eqn:exp-Y}
	\[
	\bEE{\abs{\mu_j-\widehat{\mu}_j}^{\prm}} \leq 
	\Paren{\frac{\prm}{(\ns/6\dims)(2\alpha-1)^2}}^{\prm/2}.
	\]
	Upon summing over $j\in \setS_\numT$, we obtain
	\[
	\bEE{\sum_{j \in \setS_\numT} \abs{\mu_j - \widehat{\mu}_j}^\prm} 
	\leq 
	3 \sprs\cdot
	\Paren{\frac{6(e+1)^2\dims}{(e-1)^2\ns\priv^2}}^{\prm/2} \leq 
	3 \cdot 6^\prm\cdot \sprs \Paren{\frac{\prm\dims}{\ns 
				\priv^2}}^{\prm/2}.\qedhere
	\]
\end{proof}

%% file: sec-bernoulli-ub-sparse-comm.tex
\subsubsection{Communication constraints} 
\label{sec:communication-upper}
In~\cref{alg:communication} we propose a protocol to estimate the mean of 
product Bernoulli distributions under $\numbits$-bit communication 
constraints. \newchange{As mentioned in the previous subsection, the $\priv$-LDP algorithm with $\priv>1$ will follow from a simple modification of the communication-constrained one; we discuss how
to privatize the latter to obtain the former at the end of the section.} As 
in the LDP case when $\priv \in (0, 1]$, 
in~\ref{stage-one-begin}--\ref{stage-one-end} the algorithm iteratively 
prunes an initial set $S_0=[\dims]$ to obtain a set $S_T$ of size 
$\max\{3s,\numbits\}$, which denotes the set of potential non-zero 
coordinates. We then estimate the mean of coordinates in $S_T$. If 
$\numbits>3s$, then we can directly send the values of all coordinates in 
$S_T$ and use it for estimation; otherwise, when $3s>\numbits$, we again 
partition $S_T$ into sets of size $\numbits$ and each player sends the bits 
of its sample in this set. This is done in 
Lines~\ref{stage-two-begin}--\ref{stage-two-end}. We state the 
performance of~\cref{alg:communication} below.

\begin{algorithm}[h]
	\caption{$\numbits$-bit protocol for estimating product of Bernoulli 
		family}
	\label{alg:communication}
	\begin{algorithmic}[1]
		\Require $\ns$ players, dimension $\dims$, sparsity parameter 
		$\sprs$, communication bound $\numbits$.
		\State Set $\numT \eqdef \log_3 (\dims/\max \{3\sprs, \numbits\}) 
		$, $\setS_0 \eqdef [\dims]$, $N_0 \eqdef
		\frac{\ns \numbits}{18\dims}$. 
		\For {$t = 1, 2, \ldots, \numT$}\label{stage-one-begin}
		\State Set $P \eqdef \frac{\dims}{3^{t-1}\numbits}$, and partition $S_{t - 1}$ into $P$ 
		subsets $S_{t - 1, 1}, \ldots, S_{t - 1, P}$, each of size $\numbits$. 
		\For{$j=1, 2,\ldots,  P$}
		\State Get a group of new players 
			$G_{t,j}$ of size $N_t = N_0 \cdot 2^t$. 
		\State Player $i\in G_{t,j}$, upon observing $X_i\in\bool^\dims$ sends the message $Y_i = \{(X_i)_x\}_{x \in S_{t - 1, j}}$.
		\State For  $x \in S_{t - 1, j}$, 
		let $M_{t,x} \eqdef \sum_{i\in G_{t,j}} (X_i)_x$. 
		\EndFor
		\State Set $S_{t} \subseteq S_{t-1}$ to be the 
		set of indices with the largest $|M_{t,x}|$ and $|S_{t}| = |S_{t-1}|/3$.
		\EndFor
		\If{$\numbits \le 3\sprs$}
		\State Partition $S_{\numT}$ into $3 \sprs/\numbits$ 
		subsets of size $\numbits$ each, $S_{\numT, j}, j \in [3 \sprs/\numbits]$.\label{stage-one-end}
		\For{$j=1,\ldots,3s/\numbits$}\label{stage-two-begin}
		\State Get a new group 
		$G_{\numT + 1,j}$ of players of size $\ns\numbits/(6\sprs)$. 
		\State Player 
		$i\in 
		G_{\numT + 1, j}$,
		sends 
		the message $Y_i = \{(X_i)_x\}_{x \in S_{\numT, j}}$.
		\State For $x \in S_{\numT, j}$, let 
		$M_{\numT+1,x} = \sum_{i\in 
			G_{\numT + 1,j}} (X_i)_x $. Set 
		\[
		\widehat{\mu}_x \eqdef \frac{6 \sprs}{\ns \numbits}
		M_{\numT+1, x},
		\]
		\State For $x \notin  S_{\numT}$, set $\widehat{\mu}_x= 0$.
		\EndFor
		\EndIf
		\If{$\numbits > 3\sprs$},
		\State Get $\ns/2$ new players $G_{\numT+1}$ and for $i \in G_{\numT+1}$, 
		player $i$ sends $Y_i = \{(X_i)_x\}_{x \in S_{\numT}}$\Comment{This 
			can be done since $|S_{\numT}| = \numbits$ if $\numbits > 3\sprs$.}
		\State For $x \in S_{\numT}$, let $M_{\numT+1, x} = \sum_{i\in 
			G_{\numT + 1,j}} (X_i)_x $. Set $S_{\numT+1} \subseteq 
		S_{\numT}$ 
		to be the 
		set of indices with the largest $|M_{\numT+1, x}|$ and $|S_{\numT+1}| 
		= 
		3\sprs$.
		For all $x \in  S_{\numT+1}$, set 
		\[
		\widehat{\mu}_x \eqdef \frac{2}{\ns}
		M_{\numT+1, x},
		\]
		and for all $x \notin  S_{\numT+1}, \widehat{\mu}_x = 0$.\label{stage-two-end}
		\EndIf
		\State \Return $\widehat{\mu}$.
	\end{algorithmic}
\end{algorithm}

\begin{proposition}
	\label{theorem:mean:estimation:bernoulli:ub:comm}
	Fix $\prm\in[1,\infty]$. For $\ns \geq 1$ and $\numbits \le \dims$, we 
	have 
	\cref{alg:communication} is an $(\ns, \dst)$-estimator using 
	$\cW_\numbits$ under 
	$\numbits_\prm$ loss for 
	$\mathcal{B}_{\dims,\sprs}$ with 
	$
	\dst = \bigO{\sqrt{\frac{\prm \dims \sprs^{2/\prm}}{\ns \numbits} + 
	\frac{(\prm + \log(2\numbits/\sprs)) \sprs^{2/\prm} }{\ns}} }
	$ for $\prm \le 2 \log \sprs$ and $
	\dst = \bigO{ \sqrt{\frac{\dims \log \sprs  }{\ns \numbits} + \frac{\log\numbits}{\ns}} }
	$ for $\prm >  2 \log \sprs$.
\end{proposition}

When $\numbits \le 3 \sprs$, the bound we get is $\dst \lesssim \sqrt{\frac{\prm 
\dims \sprs^{2/\prm}}{\ns \numbits} }$.
The analysis is almost identical to the case under 
LDP constraints, since in both cases, the information we get about
coordinate $j$ are samples from a Rademacher distribution with mean 
$(2\alpha - 1)\mu_j$. There are only two 
differences. (i) 
$\alpha = 
1$ instead of 
$\Theta\Paren{\priv^2}$. (ii) There is a factor of $\numbits$ more players in the 
corresponding groups. Combing both factors, we can obtain the desired 
bound by replacing 
$\priv^2$ by $\numbits$. We omit the detailed proof in this case.

When $\numbits > 3\sprs$, after $\numT \asymp \log(\dims/\numbits)$ rounds, we can 
find a subset 
$S_\numT$ 
of size $\numbits$ which contains most of the coordinates with large biases. The 
protocol then asks new players to send all coordinates within $S_\numT$ 
using $\numbits$ bits. 
In this case, it would be enough to prove 
\cref{lem:err:withinS:comm} since for the coordinates outside $S_\numT$, 
we 
can show the error is small following exactly the same steps as the proof 
for bouding the first term in \cref{lem:err:withinS} as we explained in the 
case when $\numbits \le 3 
\sprs$.
\begin{lemma}\label{lem:err:withinS:comm}
	Let $\setS_\numT$ be the subset obtained from the first stage of 
	\cref{alg:communication}, we 
	have
	\[
	\bEE{\sum_{j \in \setS_\numT} \abs{\mu_j - \widehat{\mu}_j}^\prm} = 
	O\Paren{\sprs \Paren{\frac{\prm + \log\frac{2\numbits}{\sprs}}{\ns}}^{\prm/2}}.
	\]
	\end{lemma}

\begin{proof} Similar to~\cref{lem:err:withinS}, we will 
prove that the
statement is true for any realization of $S_\numT$, which is a stronger 
statement than the claim.
	\begin{align*}
	\bEE{\sum_{j \in \setS_\numT} \abs{\mu_j - \widehat{\mu}_j}^\prm}  
	& = \bEE{ \sum_{j \in \setS_\numT}  \abs{\mu_j - \widehat{\mu}_j}^\prm \indic{j \in S_{\numT + 1}}} + \bEE{ \sum_{j\in \setS_\numT}   |\mu_j|^\prm \indic{j \notin S_{\numT + 1}}}  \\
	& \le \bEE{ \sum_{j \in \setS_{\numT + 1}}  \abs{\mu_j 
			- \widehat{\mu}_j}^\prm } + \sum_{j \in \setS_\numT} |\mu_j|^\prm \bPr{j \notin  \setS_{\numT + 1}}.
	\end{align*}
	Fix $\setS_{\numT + 1}$. For each $j \in \setS_{\numT + 1}$, 
	$M_{\numT+1, j}$ is binomially distributed with mean $\mu_j$ 
	and $\ns/2$ trials. By 
	similar 
	computations as~\cref{lem:err:withinS}, we have 
	
	\begin{equation} \label{eqn:err:final:set}
		\bEE{ \sum_{j \in \setS_{\numT + 1} } 
		\abs{\mu_j 
				- \widehat{\mu}_j}^\prm } = 
				O\Paren{\sprs 
				\Paren{\frac{\prm}{\ns}}^{\prm/2}}.
	\end{equation}
	Next we show for all $j \in \setS_{\numT}$ such that $\mu_j \neq 0$,
	\begin{equation}\label{eqn:missing:2}
	|\mu_j|^\prm \bPr{j \notin  \setS_{\numT + 1}} \leq
		2\Paren{\frac{\prm \lor  64\ln\frac{2\numbits}{\sprs}}{\ns}}^{\prm/2}.
	\end{equation}
	If $|\mu_j| \le 
	\Th' \eqdef 8\sqrt{\frac{\ln\frac{2\numbits}{\sprs} }{\ns 
			}}$,~\cref{eqn:missing:2} always holds since 
	$\bPr{j 
		\notin \setS} \le 1$. Hence we 
	hereafter assume that $|\mu_j| > 
	\Th'$, and write $\mu_j = \beta_j \Th'$ 
	for some $\beta_j  > 1$. 	
	
	Let $\numlc_{\numT+1,j}$ be 
	the 
	number 
	of coordinates $j'$ with $\mu_{j'} = 0$ and 
	$|M_{\numT+1,j'}| \ge \frac{\ns}{2}\cdot  \frac{\beta_j\Th'}{2}$. 
	Then 
	since $\setS_{\numT+1}$ contains the top $3\sprs$ coordinates with 
	the largest 
	magnitude of the sum,
	we have 
	$j \notin \setS_{\numT+1}$ happens 
	only if at least one of the 
	following occurs (i)~$\numlc_{\numT+1,j} > 2\sprs$, or (ii)~$M_{\numT+1,j} < \frac{\ns}{2}\cdot  \frac{\beta_j\Th'}{2}$. 
	
	By Hoeffding's inequality, we have
	\[
	\bPr{M_{\numT+1,j} < \frac{\ns}{2}\cdot  \frac{\beta_j\Th'}{2}} 
	\le 
	\exp\Paren{ - \frac12\cdot \frac{\ns}{2}
		\cdot\Paren{\frac{ 
				\beta_j\Th'}{2}}^2} = 
				\Paren{\frac{2\numbits}{\sprs}}^{-4\beta_j^2} 
	\eqdef p_{\numT+1,j}.
	\]	
	Similarly, for any $j'$ such that $\mu_{j'} = 0$,  
	\[
	\bPr{|M_{\numT+1,j'}| \ge \frac{\ns}{2}\cdot  \frac{
	\beta_j\Th'}{2}} \le  
	2 p_{\numT+1,j}.
	\]
	Since all coordinates are independent, $\numlc_{\numT+1,j}$ is binomially distributed with mean at most $2 p_{\numT+1,j} 
	\numbits$, and therefore, by Markov's inequality,
	\[
	\bPr{\numlc_{\numT+1,j}> 2\sprs} \le 
	\frac{2 p_{\numT+1,j} 
	\numbits }{2\sprs} \le \Paren{\frac{2\numbits}{\sprs}}^{1 -4\beta_j^2} 
	\le \Paren{\frac{2\numbits}{\sprs}}^{-3\beta_j^2} 
	\]
	the last step since $\beta_j > 1$. By a union bound, we have
	\[
		 \bPr{j \notin \setS_\numT}  \le \bPr{\numlc_{\numT+1,j}> 2\sprs} + 
		 \bPr{M_{\numT+1,j} < \frac{1}{4}\frac{\ns}{2}\cdot  \frac{\beta_j\Th'}{2}}  \le 2 
		 \Paren{\frac{2\numbits}{\sprs}}^{-3\beta_j^2}.
	\]
	Using the inequality $x^\prm a ^{-x^2} \le \Paren{\frac{\prm}{2e\ln 
	a}}^{\prm/2}$ which holds for all $x > 0$, we get overall
	\[
	|\mu_j|^\prm \cdot \bPr{j \notin \setS_\numT} 
	\leq 2 H'^\prm \beta_j^{\prm} 
	\Paren{\frac{2\numbits}{\sprs}}^{-4\beta_j^2}
	\leq 2\Paren{\frac{\prm }{e\ns}}^{\prm/2},
	\]
	establishing \cref{eqn:missing:2}.
	Combining \cref{eqn:err:final:set} and \cref{eqn:missing:2} concludes the proof~\cref{lem:err:withinS:comm} since there are at most 
	$\sprs$ unbiased coordinates.
\end{proof}

\newchange{
	\paragraph{Algorithm under LDP with $\priv > 1$} To get a
	$\priv$-LDP
	algorithm in the regime $\priv > 1$ (low-privacy regime), we perform the following changes to obtain a 
	private algorithm from \cref{alg:communication}:
	\begin{itemize}
		\item Each user  independently flips each 
		coordinate of 
		their local sample to get $Z_i$ where, for all $x \in [d]$, $(Z_i)_x = 
		(X_i)_x$ with probability $\frac{e}{e + 1}$ and $(Z_i)_x = 1 -
		(X_i)_x$ with probability $\frac{1}{e + 1}$ (note that this corresponds to applying Randomized Response independently to each bit with privacy parameter $1$).  
		\item Users then follow~\cref{alg:communication} with the setting $\numbits = 
		\flr{\priv}$ and local data $\{ Z_i \}_{i \in [\ns]}$, and obtain estimate 
		$\widehat{\mu}$.
		\item The final estimate is then $\frac{e + 1}{e -  1}\widehat{\mu}$.
	\end{itemize} 
The privacy guarantee of the algorithm comes from the fact that  
\cref{alg:communication}  sends at most $\numbits = \flr{\priv}$ coordinates of 
each $Z_i$, and for any $S$ with $|S| \le \flr{\priv}$ 
\[
\frac{\bPr{\{(Z_i)_x\}_{x \in S} \mid X_i}}{\bPr{\{(Z_i)_x\}_{x \in S} \mid X'_i}} 
= \prod_{x \in S} \frac{\bPr{(Z_i)_x \mid (X_i)_x}}{\bPr{(Z_i)_x \mid (X'_i)_x}} 
\le e^{\flr{\priv}}.
\]
The utility guarantee follows from observing that $\mu_Z = \frac{e - 1}{e + 1}\mu$ and 
hence any $\lp[\prm]$ error guarantee will be preserved up to a constant.
}

%% file: gaussian-ub.tex
Recall that $\mathcal{G}_{\dims,\sprs}$ denotes the family of 
$\dims$-dimensional spherical Gaussian distributions with 
$\sprs$-sparse mean in 
$[-1,1]^\dims$, \ie{}
\begin{equation}\label{def:sparse:gaussian}
    \mathcal{G}_{\dims,\sprs} = \setOfSuchThat{\gaussian{\mu}{\II}}{ 
    \norminf{\mu} \leq 1, \norm{\mu}_0 \leq \sprs}\,.
\end{equation}
We will prove the following results for LDP and communication constraints, respectively.
\begin{proposition}
	\label{theorem:mean:estimation:gaussian:ub:ldp}
	Fix $\prm\in[1,\infty]$. For $\ns \geq 1$ and $\priv \in (0,1]$, there 
	exists 
	an $(\ns, \dst)$-estimator using $\cW_\priv$ under 
	$\numbits_\prm$ loss for 
	$\mathcal{G}_{\dims,\sprs}$ with 
	$
	\dst = \bigO{\sqrt{\frac{\prm \dims \sprs^{2/\prm}}{\ns \priv^2}}}
	$ for $\prm \le 2 \log \sprs$ and $
	\dst = \bigO{\sqrt{\frac{\dims \log \sprs  }{\ns \priv^2}}}
	$ for $\prm >  2 \log \sprs$.
\end{proposition}
\begin{proposition}
	\label{theorem:mean:estimation:gaussian:ub:comm}
	Fix $\prm\in[1,\infty]$. For $\ns \geq 1$ and $\numbits \le \dims$, 
	there exists an $(\ns, \dst)$-estimator using 
	$\cW_\numbits$ under 
	$\numbits_\prm$ loss for 
	$\mathcal{G}_{\dims,\sprs}$ with 
	$
	\dst = \bigO{\sqrt{\frac{\prm \dims \sprs^{2/\prm}}{\ns \numbits} + 
		\frac{(\prm + \log(2\numbits/\sprs)) \sprs^{2/\prm} }{\ns}}}
	$ for $\prm \le 2 \log \sprs$ and $
	\dst = \bigO{\sqrt{\frac{\dims \log \sprs  }{\ns \numbits} + \frac{\log \numbits}{\ns}}}
	$ for $\prm >  2 \log \sprs$.
\end{proposition}

We reduce the problem of Gaussian mean estimation to that of Bernoulli 
mean estimation and then 
invoke
\cref{theorem:mean:estimation:bernoulli:ub:ldp,theorem:mean:estimation:bernoulli:ub:comm}
 from the previous section. 
	At the heart of the reduction is a simple idea that was used in, 
	\eg{}~\cite{BGMNW:16,ACT:20,CKMUZ:19}: the sign of a Gaussian random 
	variable already preserves sufficient information about the mean. Details 
	follow.

Let $\p\in\mathcal{G}_{\dims, \sprs}$ with mean 
$\mu(\p)=(\mu(\p)_1,\ldots, \mu(\p)_\dims)$. 
For $X\sim\p$, let $Y = (\sign(X_i))_{i\in[\dims]}\in \bool^\dims$ be a random variable indicating the signs of the $\dims$ coordinates of $X$. By the independence of the coordinates of $X$, note that $Y$ is distributed as a product Bernoulli distribution (in $\mathcal{B}_\dims$) with mean vector $\nu(\p)$ given by
\begin{equation}
    \label{eq:def:nup}
    \nu(\p)_i = 2\bP{X\sim\p}{ X_i > 0}-1 = \operatorname{Erf}\Paren{\frac{\mu(\p)_i}{\sqrt{2}}}, \qquad i\in[\dims],
\end{equation}
and, since $|\mu(\p)_i|\le 1$, we have $\nu(\p)\in[-\eta,\eta]^\dims$, where $\eta\eqdef \operatorname{Erf}\Paren{1/\sqrt{2}}\approx 0.623$. 
Moreover, it is immediate to see that each player, given a sample from $\p$, can convert it to a sample from the corresponding product Bernoulli distribution. 
We now show that a good estimate for $\nu(\p)$ yields a good estimate for $\mu(\p)$.
\begin{lemma}
  \label{lemma:mean:estimation:gaussian:ub:reduction}
    Fix any $\prm\in[1,\infty)$, and $\p\in\mathcal{G}_\dims$. For $\widehat{\nu}\in[-\eta,\eta]^\dims$, define $\widehat{\mu}\in[-1,1]^\dims$ by
    $
        \widehat{\mu}_i \eqdef \sqrt{2}\operatorname{Erf}^{-1}(\widehat{\nu}_i),
    $ for all $i\in[\dims]$. 
    Then
    \[
        \norm{\mu(\p)-\widehat{\mu}}_\prm \leq \sqrt{\frac{e\pi}{2}}\cdot \norm{\nu(\p)-\widehat{\nu}}_\prm\,.
    \]
\end{lemma}
\begin{proof}
  By computing the maximum of its derivative,\footnote{Specifically, we have that $\max_{x\in[-\eta,\eta]} \operatorname{Erf}^{-1}(x) = 1/\sqrt{2}$ by definition of $\eta$ and monotonicity of $\operatorname{Erf}$. Recalling then that, for all $x\in[-\eta,\eta]$, $(\operatorname{Erf}^{-1})'(x) = \frac{1}{\operatorname{Erf}'(\operatorname{Erf}^{-1}(x))}= \frac{\sqrt{\pi}}{2}e^{(\operatorname{Erf}^{-1}(x))^2} \leq \frac{\sqrt{\pi}}{2} e^{\frac{1}{2}}$, we get the Lipschitzness claim.} we observe that the function $\operatorname{Erf}^{-1}$ is $\frac{\sqrt{e\pi}}{2}$-Lipschitz on $[-\eta,\eta]$. 
  By the definition of $\widehat{\mu}$ and recalling~\cref{eq:def:nup}, we then have
  \begin{align*}
      \norm{\mu(\p)-\widehat{\mu}}_\prm^\prm
        &= \sum_{i=1}^\dims \abs{\mu(\p)_i-\widehat{\mu}_i}^\prm
        = 2^{\prm/2}\cdot\sum_{i=1}^\dims \abs{\operatorname{Erf}^{-1}(\nu_i)-\operatorname{Erf}^{-1}(\widehat{\nu}_i)}^\prm
        \leq \Paren{\frac{e\pi}{2}}^{\prm/2} \cdot\sum_{i=1}^\dims \abs{\nu_i-\widehat{\nu}_i}^\prm,
  \end{align*}
  where we used the fact that $\nu,\widehat{\nu}\in[-\eta,\eta]^\dims$.
\end{proof}
As previously discussed, 
combining~\cref{lemma:mean:estimation:gaussian:ub:reduction} with~
\cref{theorem:mean:estimation:bernoulli:ub:ldp,theorem:mean:estimation:bernoulli:ub:comm}
 (with $\dst' \eqdef \sqrt{\frac{2}{e\pi}}\dst$) immediately 
implies~\cref{theorem:mean:estimation:gaussian:ub:ldp,theorem:mean:estimation:gaussian:ub:comm}
 for $\prm \in[1,\infty]$. 

\begin{remark}
	Note that for the Gaussian family, we also consider the linear measurement 
	constraint. Under linear measurement constraints, we can use the linear 
	measurement matrix to obtain $\ldim$ out of $\dims$ coordinates and 
	perform the above reduction to product of Bernoulli family. The obtained 
	bound will be same as that under communication constraints.
\end{remark}